\documentclass[12pt]{article}

\usepackage{latexsym,amssymb,amsfonts,amsmath,amsthm,braket,appendix}
\usepackage{mathtools}
\usepackage{ytableau}
\usepackage{enumitem}
\usepackage{graphicx}
\usepackage{comment}
\usepackage{subcaption}

\usepackage[normalem]{ulem}
\usepackage{color}

\newtheorem{theorem}{Theorem}
\newtheorem{lemma}{Lemma}

\newtheorem{proposition}{Proposition}
\newtheorem{remark}{Remark}

\numberwithin{theorem}{section}
\numberwithin{lemma}{section}
\numberwithin{corollary}{section}
\numberwithin{proposition}{section}
\numberwithin{remark}{section}
\numberwithin{definition}{section}

\title{Pulsation of quantum walk on finite graphs}
\author{Taisuke Hosaka$^{1}$\thanks{Corresponding Author: hosaka-taisuke-xh@ynu.ac.jp },  Etsuo Segawa$^1$\\
$^1${\small Faculty of Environment and Information Sciences, Yokohama National University,}\\ 
{\small Hodogaya, Yokohama, 240-8501, Japan,
}\\}
\date{}

\begin{document}
\maketitle
\centerline{\textbf{Abstract}}
We study discrete-time quantum walks on weighted graphs, where every edge in an edge cut set is assigned a small weight $\epsilon>0$.
The parameter $\epsilon$ represents the strength of the connections through the cut edges: as $\epsilon \to 0$, these connections vanish, and the graph decomposes into the connected components obtained by removing the cut edges. 
We refer to these weighted cut edges as {\it weak edges}. 
We show that, for sufficiently small $\epsilon$ and on the time scale $t=\Theta(\epsilon^{-1/2})$, the finding probabilities between the connected components are asymptotically described by a continuous-time wave equation on a reduced graph, whose vertices represent the connected components and whose edges represent the weak edges connecting them. The wave equation is governed by a symmetric weighted Laplacian determined by the reduced graph and the number of arcs contained in the components. Consequently, finding probabilities of leading-term are independent of the detailed internal structures of the components.
We further show that this wave equation has the same form as Newton's equation of motion for a classical spring-mass system on the reduced graph.

\mbox{} 

\section{Introduction}

Quantum walks have been studied as a useful framework for quantum information and quantum algorithms
\cite{ambainis2007,ambainis2001,grover1996}.
Their behavior is quite different from that of classical random walks, leading to various characteristic phenomena such as localization, ballistic spreading, and periodicity
\cite{grimmett2004,higuchi2017,inui2005,konno2002}.
Among many topics in quantum walks, quantum search and perfect state transfer have been studied, which are relative to quantum algorithms and quantum communication, originating from for example, \cite{shenvi2003, ambainis2005} and \cite{godsil2012a}, respectively. 
Such algorithms are driven by quantum walk models on graphs that start from a specific state (mainly the uniform state) or a specific vertex, and reach a specific set of vertices with a high probability at a short time. We note that once the time step passes this moment, and the same amount of time has elapsed, it will return to the original state. 
Then, it is natural to generalize the above target by the movement back and forth between different regions on a graph and to study its behavior, in particular, when its probability becomes concentrated in a particular region.
Those behaviors are strongly influenced by the structure of the underlying graph structure of course, and in particular, the characterization by its spectral property has been investigated, for example,  
\cite{godsil2012, godsil2019, kubota2022, GodsilHanmengbook, AdaChanHanmeng}.
On the other hand, in this paper, we pay attention to how the underlying graph consists of several subgraphs connected through a relatively small number of edges and we regard these subgraphs as the regions.
Quantum walks on graphs in this viewpoint have been studied in connection with quantum transport and the role of connections between different regions \cite{mukai2020,tsomokos2011}.
Weighted quantum walks also provide a natural framework in which the strength of transitions along individual edges can be varied \cite{wong2017,wong2016}.
These studies motivate us to investigate how weak connections between different regions affect the transport of a quantum walker.

We consider a graph consisting of several subgraphs connected through edges whose weights are much smaller than those of the internal edges.
From a physical viewpoint, such weak connections may be able to viewed as quantum tunneling between different regions: transitions across them are strongly suppressed, while transport between the regions remains possible over sufficiently long time scales.
This raises the question whether the transport between the subgraphs can be described without analyzing their detailed internal structures.
A natural approach is to contract each subgraph into a single vertex and study the resulting reduced graph.

In our previous work \cite{hosaka2026}, we considered a simple setting consisting of two graphs connected by a single edge with a small weight and showed that the quantum walker can move repeatedly between the two subgraphs, a behavior that we call the {\it pulsation phenomenon} \cite{hosaka2025}.
This result motivates us to extend the analysis to more general graphs consisting of several regions connected through weak edges.
Consequently, we consider a finite weighted graph in which every edge belonging to an edge cut set is assigned a small weight $\epsilon>0$, while all remaining edges have weight $1$.
We refer to these edges as {\it weak edges}.
As $\epsilon$ decreases, transitions across weak edges are increasingly suppressed, and at $\epsilon=0$ the graph separates into the connected components obtained by removing them.
We study the transport between these components for sufficiently small $\epsilon$.

We take as the initial state the uniform superposition of the arc states contained in one connected component.
Our main result shows that, on the time scale $t=\Theta(\epsilon^{-1/2})$, the finding probabilities between the connected components are asymptotically described by a continuous-time wave equation on a reduced graph (Theorem~\ref{thm:prob_matrix}).
The reduced graph is obtained by contracting each connected component into a single vertex while retaining the weak edges that connect the components.
The wave equation is governed by a symmetric weighted Laplacian $L_{\mathrm{sym}}$, determined by the numbers of weak edges between the components and the number of arcs contained in the individual components.
Thus, although the original quantum walk evolves in the full arc space of the graph, the asymptotic transport between the components is described by an operator in the much smaller vertex space of the reduced graph.
In particular, the leading asymptotic term of the component-wise finding probabilities depends only on the reduced graph and the numbers of arcs in the components, and not on their detailed internal structures.

The wave equation in the main theorem also admits a classical mechanical interpretation.
We show that it has the same form as Newton's equation of motion for a spring-mass system on the reduced graph.
In this correspondence, the number of arcs contained in each connected component plays the role of its mass, whereas the number of weak edges connecting two components plays the role of the spring constant between them.
Hence, the spring-mass system provides a classical mechanical interpretation of the asymptotic transport described by our quantum-walk model.

This paper is organized as follows.
Section~\ref{subsec_setting_many} introduces the model and notation.
In Section~\ref{subsec_theorem_many}, we state the main result.
Section~\ref{subsec_simu} presents numerical simulations and the spring-mass interpretation of the main result.
In Section~\ref{subsec_proof_many}, we prove the main result.
Section~\ref{subsec_summary} summarizes and discusses the results.

\section{Setting the model}
\label{subsec_setting_many}
Let $G=(X,A)$ be a finite, simple and connected graph, where $X$ is the set of vertices and $A$ is the set of symmetric arcs. 
For each $a \in A$, its origin and terminus are denoted by $o(a)$ and $t(a)$, respectively. 
Each arc $a$ has an inverse arc $a^{-1} \in A$, satisfying $o(a)=t(a^{-1})$ and $t(a)=o(a^{-1})$.
Let $B \subset A$ be a symmetric arc set corresponding to an edge cut set such that removing the arcs in $B$ separates $G$ into connected components.
Remark that $a \in B \Longleftrightarrow a^{-1} \in B$ holds. We refer to such a set $B$ as a \textit{weak edge set}. 
This terminology is motivated by the weight function $w_\epsilon$ as follows:
For $\epsilon \geq 0$, the weight function $w_{\epsilon}: A \to [0, \infty)$ is defined by
\begin{align}
    \label{eq: weight}
    w_{\epsilon}(a) =
    \begin{dcases}
        \epsilon &: a \in B, \\
        1 &: a \notin B.
    \end{dcases}
\end{align}
Since $B$ is symmetric, $w_{\epsilon}(a) = w_{\epsilon}(a^{-1})$ holds for all $a \in A$. 
We consider \(G\) as a weighted graph with the weight function \(w_\epsilon\) defined in Eq.~\eqref{eq: weight}. The parameter \(\epsilon>0\) represents the strength of the connections through the weak edges. As $\epsilon\to0$, the weights of the arcs in $B$ vanish. In particular, at $\epsilon=0$, the weighted graph separates into the connected components obtained by removing $B$.

Assume that the graph with the vertex set $X$ and the arc set $A\setminus B$ consists of $n$ connected components, denoted by $G_{1}=(X_{1}, A_{1}), \dots, G_{n}=(X_{n}, A_{n})$, where 
$A_{j}=\{a\in A\setminus B : o(a),t(a)\in X_{j}\}$.
We assume that each connected component contains at least one edge; equivalently, $|A_j|>0$ for every $j=1,\ldots,n$.
Then, we see $X=\bigsqcup_{j=1}^{n} X_{j}$ where $\bigsqcup$ denotes the disjoint union.
Let $\pi: X \to \tilde{X}=\{X_{1}, \dots, X_{n}\}$ be the natural surjection associated with the partition of the set of vertex $X$, that is, $\pi(x)=X_{j}$ if $x \in X_{j}$. 
We define the reduced graph by
\begin{align*}
    \tilde{G}=(\tilde{X},B),
\end{align*}
where each arc $a\in B$ is regarded as an arc between the corresponding connected components.
For convenience, we identify each vertex $V\in\tilde X$ with the vertex set of the corresponding connected component of $G$.
Thus, $V$ is used both as a vertex of the reduced graph and as the corresponding subset of $X$.
In general, $\tilde{G}$ is a symmetric multigraph that may contain
multiple arcs and self-loops.
For $V,W \in \tilde{X}$, $\ell(V,W)$ is the number of directed weak edges from $V$ to $W$ in the original graph $G$, denoted by
\begin{align*}
    \ell(V,W)=\left|\left\{ a\in B:\pi(o(a))=V,\pi(t(a))=W \right\}\right|.
\end{align*}
Since $B$ is symmetric, $\ell(V,W)=\ell(W,V)$ holds.
For each $x\in X$, we define
\begin{align*}
    \deg(x) = |\{a \in A\setminus B \mid o(a)=x\}|, \quad \deg_{B}(x) = |\{a \in B \mid o(a)=x\}|.
\end{align*}
Thus, $\deg(x)$ counts the arcs in $A\setminus B$ originating from $x$, whereas $\deg_B(x)$ counts the weak arcs originating from $x$.
For each $V\in\tilde{X}$, its degree in the reduced graph $\tilde{G}$ is given by
\begin{align*}
\deg_{\tilde{G}}(V)=\sum_{x\in V}\deg_B(x)=\sum_{W\in\tilde{X}}\ell(V,W).
\end{align*}
The weighted degree of $x\in X$ is defined by
\begin{align*}
    m_{\epsilon}(x) = \sum_{a: o(a)=x} w_{\epsilon}(a)=\deg (x)+\epsilon \deg_{B}(x).
\end{align*}
We then define the transition probability, which gives the probability that a random walker on $o(t)$ moves along the arc $a$ on the original graph $G$ by
\begin{align*}
    p_{\epsilon}(a) = \frac{w_{\epsilon}(a)}{m_{\epsilon}(o(a))}.
\end{align*}
Then $p_\epsilon(a)$ represents the transition probability for the walker to move from $o(a)$ to $t(a)$ along the arc $a$. 
In particular, the transition probabilities along the weak arcs decrease as $\epsilon$ becomes smaller, so that $\epsilon$ controls the coupling between the components $G_1,\ldots,G_n$.

We introduce the boundary operator $d_{\epsilon}: \mathbb{C}^{A} \to \mathbb{C}^{X}$ defined by
\begin{align*}
    (d_{\epsilon}\psi)(x) = \sum_{a \in A \,:\, t(a)=x} \sqrt{p_{\epsilon}(a^{-1})}\psi(a)
\end{align*}
for any $\psi \in \mathbb{C}^{A}$ and $x \in X$. 
Its adjoint $d_{\epsilon}^{*}: \mathbb{C}^{X} \to \mathbb{C}^{A}$ is given by
\begin{align*}
    (d_{\epsilon}^{*}f)(a) = \sqrt{p_{\epsilon}(a^{-1})}f(t(a))
\end{align*}
for any $f \in \mathbb{C}^{X}$ and $a \in A$. 
By direct calculation, we can verify that $d_{\epsilon}d^{*}_{\epsilon}$ is the identity operator in $\mathbb{C}^{X}$, and $d_{\epsilon}^{*}d_{\epsilon}$ is an orthogonal projection in $\mathbb{C}^{A}$. 
Furthermore, the shift operator $S: \mathbb{C}^{A} \to \mathbb{C}^{A}$ is defined as the flip-flop shift, given by
\begin{align*}
    (S\psi)(a) = \psi(a^{-1})
\end{align*}
for any $\psi \in \mathbb{C}^{A}$ and $a \in A$.

The time evolution operator $U(\epsilon): \mathbb{C}^{A} \to \mathbb{C}^{A}$ is defined by
\begin{align*}
    U(\epsilon) = S(2d_{\epsilon}^{*}d_{\epsilon} - I),
\end{align*}
where $I$ denotes the identity operator on $\mathbb{C}^A$. 
At $\epsilon=0$, transitions through weak edges are completely suppressed, and the quantum walk restricted to the internal arcs decomposes into independent Grover walks on the connected components $G_1,\ldots,G_n$.
On the other hand, if $\epsilon=1$, the time evolution corresponds to the standard Grover walk on the entire graph $G$.
Here, the time evolution operator of the Grover walk is given by
\begin{align*}
    (U(1)\psi)(a)=\frac{2}{m_{1}(t(a))} \sum_{b\,:\, o(b)=t(a)} \psi(b)-\psi(a^{-1})
\end{align*}
for $a \in A$ and $\psi\in \mathbb{C}^A$.
This expression indicates that a wave incident along an arc $a$ is transmitted to each of the other outgoing arcs with weight $2/\deg(t(a))$, and is reflected back along the reverse arc $a^{-1}$ with weight $2/\deg(t(a))-1$.
Let $\psi_{0}^{(j)} \in \mathbb{C}^{A}$ be the initial state, defined as the uniform state on a connected component $G_{j}$, given by
\begin{align}
    \label{eq: initial}
    \psi_{0}^{(j)}(a)=
    \begin{dcases}
        \frac{1}{\sqrt{|A_{j}|}} &: a \in A_{j}, \\
        0 &: \text{otherwise}.
    \end{dcases}
\end{align}
Let $\psi_{t}^{(j)} \in \mathbb{C}^{A}$ denote the state at time $t=0,1,2,...$, determined by the iteration
\begin{align}
    \label{eq: iteration}
    \psi_{t+1}^{(j)} = U(\epsilon)\psi_{t}^{(j)}=U(\epsilon)^{t+1}\psi_{0}^{(j)}.
\end{align}
Let $\mu_{t}^{(\epsilon)}$ be the $n \times n$ matrix whose $(j,k)$-entry, denoted by $\mu_{t}^{(\epsilon)}(j,k)$, represents the probability of finding the walker in the connected component $G_{k}$ at time $t$, given that the initial state is the uniform state on $G_{j}$. 
Specifically, the matrix elements are defined by
\begin{align}
    \label{eq: prob}
    \mu_{t}^{(\epsilon)}(j,k) = \sum_{a \in A_{k}} \left| \psi_{t}^{(j)}(a) \right|^{2}.
\end{align}
Note that the sum of each row in the matrix $\mu_{t}^{(\epsilon)}$ does not necessarily equal 1, since it does not account for the probability of finding the walker on the weak edges.

\section{Main Results}
\label{subsec_theorem_many}
In this section, we present our main theorems. Before stating the theorem, we introduce some necessary preliminaries. 
For each $V\in \tilde{X}$, the corresponding connected component is also denoted by its vertex set $V$, for a simplicity.
Moreover, for each $V\in\tilde{X}$, let $A_{V}$ denote the set of arcs contained in the component $V$; that is, $A_V = A_j$ if $V = X_j$.
Let $D_{B}: \mathbb{C}^{\tilde{X}} \to \mathbb{C}^{\tilde{X}}$ and $M_{B}: \mathbb{C}^{\tilde{X}} \to \mathbb{C}^{\tilde{X}}$ be the degree operator and adjacency operator of the reduced graph $\tilde{G}$, denoted by
\begin{align}
\label{eq: def_D_B}
    (D_{B}\tilde{f})(V)=\deg_{B}(V)\tilde{f}(V), \quad (M_{B}\tilde{f})(V)=\sum_{W: V\sim W}\ell(V,W)\tilde{f}(W)
\end{align}
for $V \in \tilde{X}$ and $\tilde{f} \in \mathbb{C}^{\tilde{X}}$. 
Let $L: \mathbb{C}^{\tilde{X}} \to \mathbb{C}^{\tilde{X}}$ be the Laplacian operator of the reduced graph $\tilde{G}$, defined by
\begin{align}
    \label{eq: laplacian}
    L=D_{B}-M_{B}.
\end{align}
We also define a diagonal operator $R: \mathbb{C}^{\tilde{X}} \to \mathbb{C}^{\tilde{X}}$ given by
\begin{align}
    \label{eq: def_R}
    (R\tilde{f})(V) &= \frac{1}{|A_{V}|}\tilde{f}(V)
\end{align}
for $V \in \tilde{X}$ and $\tilde{f} \in \mathbb{C}^{\tilde{X}}$. 
Let $L_{\mathrm{sym}}$ denote the symmetric weighted Laplacian operator associated with $R$, defined by
\begin{align}
    \label{eq: sym_laplacian}
    L_{\mathrm{sym}} = R^{1/2} L R^{1/2}.
\end{align}
We immediately see that 
\begin{align}
\label{eq: Lsym_x}
(L_{\mathrm{sym}})_{V,W}
=
\frac{\deg_{\tilde{G}}(V)}{|A_V|}\delta_{V,W}
-
\frac{\ell(V,W)}{\sqrt{|A_V||A_W|}}
\end{align}
for $V,W\in\tilde{X}$, where $\delta_{V,W}$ denotes the Kronecker delta.

We are now ready to state our main result, which shows that the leading-term finding probabilities of the original quantum walk can be described by a continuous-time wave equation on the reduced graph.
\begin{theorem}\label{thm:prob_matrix}
    Let $L_{\mathrm{sym}}$ be the weighted symmetric Laplacian defined by Eq.\,(\ref{eq: sym_laplacian}), and let $\mu_{t}^{(\epsilon)}$ be the finding probability matrix defined by Eq.\,(\ref{eq: prob}). 
Then, for any fixed $\tau>0$ and sufficiently small $\epsilon>0$, its $(j,k)$-entry satisfies
    \begin{align*}
        \mu_{\lfloor \tau/\sqrt{2\epsilon} \rfloor}^{(\epsilon)}(j,k) = \left| \Phi_{\tau}^{[C]}(j,k) \right|^{2} + O(\sqrt{\epsilon}).
    \end{align*}
    where \(\Phi_{\tau}^{[C]} \in \mathbb{C}^{\tilde{X}\times\tilde{X}}\) denotes the continuous-time evolution satisfying the following wave equation:
    \begin{align}
    \label{eq: wave_reduced}
        \frac{\partial^{2}}{\partial\tau^{2}}\Phi_{\tau}^{[C]}=-L_{\mathrm{sym}}\Phi_{\tau}^{[C]}, \quad 
        \Phi_{0}^{[C]}=I,\quad  \frac{\partial}{\partial \tau}\Phi_{0}^{[C]}=0.
    \end{align}
\end{theorem}
Theorem~\ref{thm:prob_matrix} shows that the leading term of the finding probability is completely determined by symmetric weighted Laplacian \(L_{\mathrm{sym}}\) on the reduced graph. Since \(L_{\mathrm{sym}}\) depends only on the weak connections between the components and their number of arcs, the leading-term of finding probabilities is independent of the detailed internal structures of the components. Thus, although the original quantum walk evolves on the full arc space \(\mathbb{C}^{A}\), its component-wise finding probabilities are asymptotically determined by an operator acting only on the vertex space \(\mathbb{C}^{\tilde X}\) of the reduced graph, where each vertex represents a connected component of the original graph.
Moreover, the matrix \(\Phi_\tau^{[C]}\) is the solution of the continuous-time wave equation generated by \(L_{\mathrm{sym}}\). In fact, it is explicitly given by
\begin{align}
\label{eq: solution}
    \Phi_\tau^{[C]}=
\cos\!\left(\tau\sqrt{L_{\mathrm{sym}}}\right).
\end{align}
Therefore, on the time scale \(t=\Theta(\epsilon^{-1/2})\), the leading-term finding probabilities of the discrete-time quantum walk are described by continuous-time wave propagation on the reduced graph. The wave equation in Theorem~\ref{thm:prob_matrix} also admits a classical mechanical interpretation. In particular, it has the same form as the equation of motion of a spring-mass system, where the number of arcs in each connected component corresponds to the mass and the weak-edge multiplicities correspond to the spring constants. We discuss this correspondence in more detail in the next section.

\begin{remark}
\label{rem: conti}
The estimate obtained in the proof, see Eq.~\eqref{eq: order_t}, also describes the behavior on shorter time scales. If $t=t(\epsilon)$ satisfies $t=o(\epsilon^{-1/2})$, then
\begin{align*}
    \lim_{\epsilon\to0}
\mu_{t(\epsilon)}^{(\epsilon)}(j,k)
=
\delta_{j,k}.
\end{align*}
Thus, on shorter time scales, the walker remains in its initial component in leading order. Therefore, $t=\Theta(\epsilon^{-1/2})$ is the first time scale on which nontrivial transport through the weak edges appears.
\end{remark}

\section{Interpretations of the main result}
\label{subsec_simu}
\subsection{Numerical simulations}
We present numerical simulations for the cases $\tilde{G}=C_3$ and
$\tilde{G}=P_3$, shown in Figs.~\ref{fig: cycle} and \ref{fig: path},
respectively.
Here, $C_3$ denotes the cycle graph on three vertices, while $P_3$
denotes the path graph on three vertices.
In both cases, each connected component $G_j$ of the original graph is
a complete graph $K_4$, and the initial state is uniformly distributed
in $G_1$.
Since $|A_j|=12$ for $j=1,2,3$, the corresponding symmetric weighted
Laplacians are given by
\begin{align*}
L_{\mathrm{sym}}^{(C_3)}
&=
\frac{1}{12}
\begin{pmatrix}
2 & -1 & -1 \\
-1 & 2 & -1 \\
-1 & -1 & 2
\end{pmatrix},
\qquad
L_{\mathrm{sym}}^{(P_3)}
=
\frac{1}{12}
\begin{pmatrix}
1 & -1 & 0 \\
-1 & 2 & -1 \\
0 & -1 & 1
\end{pmatrix}.
\end{align*}
Using the explicit formula for $\Phi_\tau^{[C]}$ in
Eq.~\eqref{eq: solution}, we obtain explicit expressions for the
finding probabilities in both cases.

For $\tilde{G}=C_3$, we obtain
\begin{align*}
\mu_{\lfloor \tau/\sqrt{2\epsilon} \rfloor}^{(\epsilon)}(1,1)
&=
\left(
\frac{1}{3}
+
\frac{2}{3}\cos\left(\frac{\tau}{2}\right)
\right)^2
+
O(\sqrt{\epsilon}),\\
\mu_{\lfloor \tau/\sqrt{2\epsilon} \rfloor}^{(\epsilon)}(1,2)
&=
\mu_{\lfloor \tau/\sqrt{2\epsilon} \rfloor}^{(\epsilon)}(1,3)\\
&=
\left(
\frac{1}{3}
-
\frac{1}{3}\cos\left(\frac{\tau}{2}\right)
\right)^2
+
O(\sqrt{\epsilon}).
\end{align*}
As seen in Fig.~\ref{fig: cycle}, the finding probabilities in $G_2$
and $G_3$ almost coincide.
This is also consistent with the explicit formulas above, which give
the same leading-term expression for
$\mu_t^{(\epsilon)}(1,2)$ and $\mu_t^{(\epsilon)}(1,3)$.
The probability initially concentrated in $G_1$ is symmetrically transferred to $G_2$ and $G_3$, and subsequently returns to
$G_1$.

On the other hand, for $\tilde{G}=P_3$, we obtain
\begin{align*}
\mu_{\lfloor \tau/\sqrt{2\epsilon} \rfloor}^{(\epsilon)}(1,1)
&=
\left(
\frac{1}{3}
+
\frac{1}{2}\cos\left(\frac{\tau}{\sqrt{12}}\right)
+
\frac{1}{6}\cos\left(\frac{\tau}{2}\right)
\right)^2
+
O(\sqrt{\epsilon}),\\
\mu_{\lfloor \tau/\sqrt{2\epsilon} \rfloor}^{(\epsilon)}(1,2)
&=
\left(
\frac{1}{3}
-
\frac{1}{3}\cos\left(\frac{\tau}{2}\right)
\right)^2
+
O(\sqrt{\epsilon}),\\
\mu_{\lfloor \tau/\sqrt{2\epsilon} \rfloor}^{(\epsilon)}(1,3)
&=
\left(
\frac{1}{3}
-
\frac{1}{2}\cos\left(\frac{\tau}{\sqrt{12}}\right)
+
\frac{1}{6}\cos\left(\frac{\tau}{2}\right)
\right)^2
+
O(\sqrt{\epsilon}).
\end{align*}
In this case, the explicit formulas involve two different frequencies,
$\tau/\sqrt{12}$ and $\tau/2$.
This leads to the more complicated oscillatory behavior observed in
Fig.~\ref{fig: path}: the probability is transferred from $G_1$
through the intermediate component $G_2$ to the opposite end component
$G_3$, and subsequently returns to $G_1$.

In both examples, the numerical results agree well with the asymptotic
formulas and the corresponding curves almost overlap.
A small discrepancy becomes more visible as time increases, which is
consistent with the time-dependent error estimate obtained in Eq.~\eqref{eq: order_t}.
Although the internal structures of the connected components are identical in the two examples, their transport behaviors are clearly different. This difference originates from the structure of the reduced graph. Therefore, these numerical results illustrate the reduced-graph description given in Theorem~\ref{thm:prob_matrix}.

\begin{figure}[htbp]
  \centering
  \begin{minipage}{0.55\columnwidth}
    \centering
    \includegraphics[width=\columnwidth]{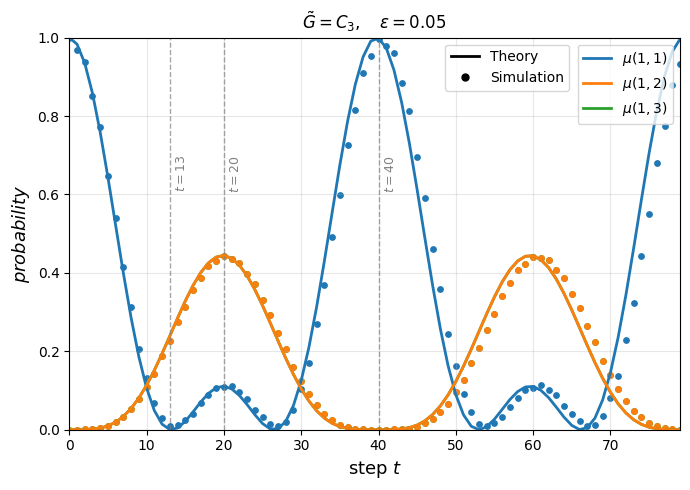}
  \end{minipage}
  \hspace{5mm}
  \begin{minipage}{0.35\columnwidth}
    \centering
    \includegraphics[width=\columnwidth]{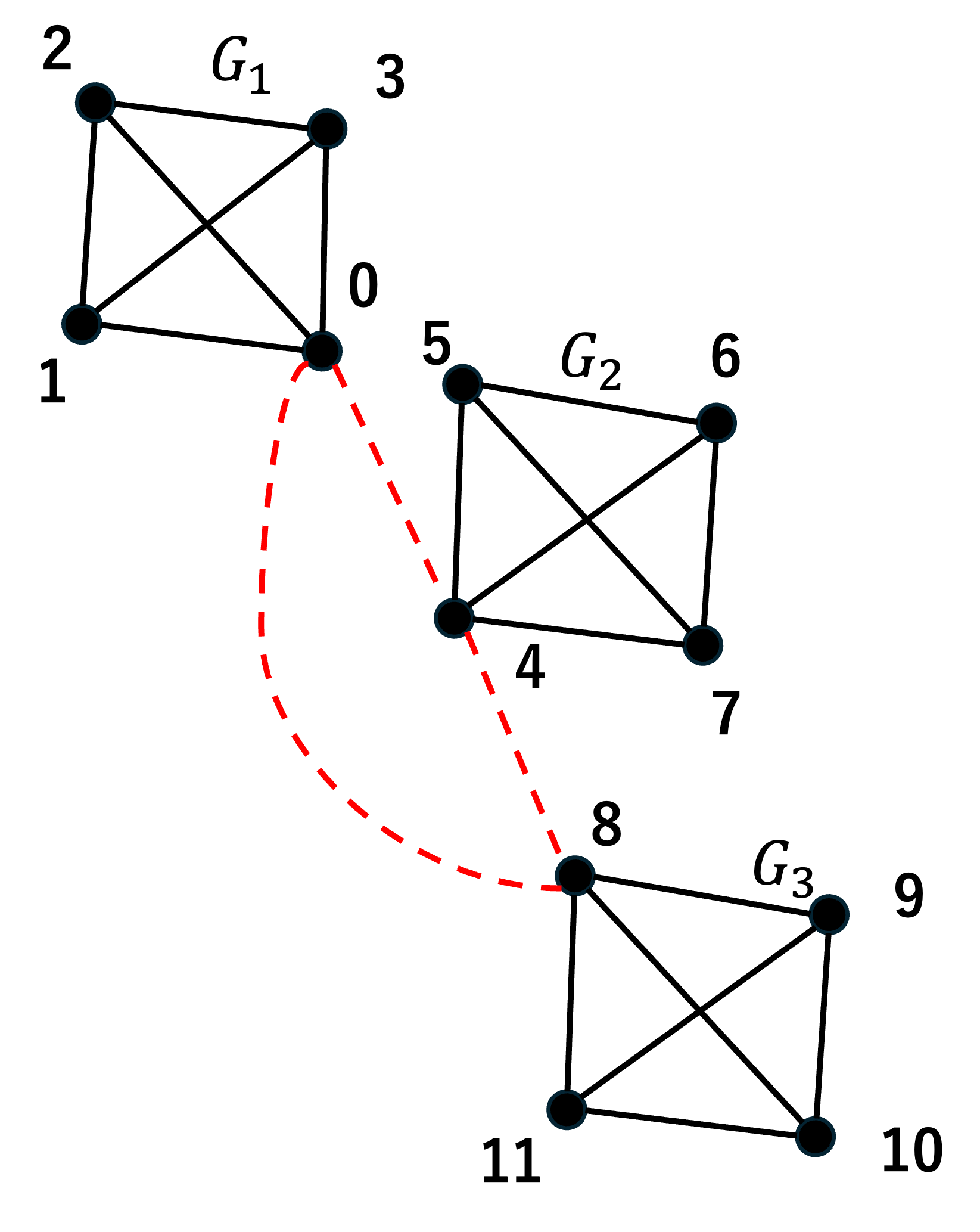}
  \end{minipage} \\

  \begin{minipage}{0.23\columnwidth}
    \centering
    \includegraphics[width=\columnwidth]{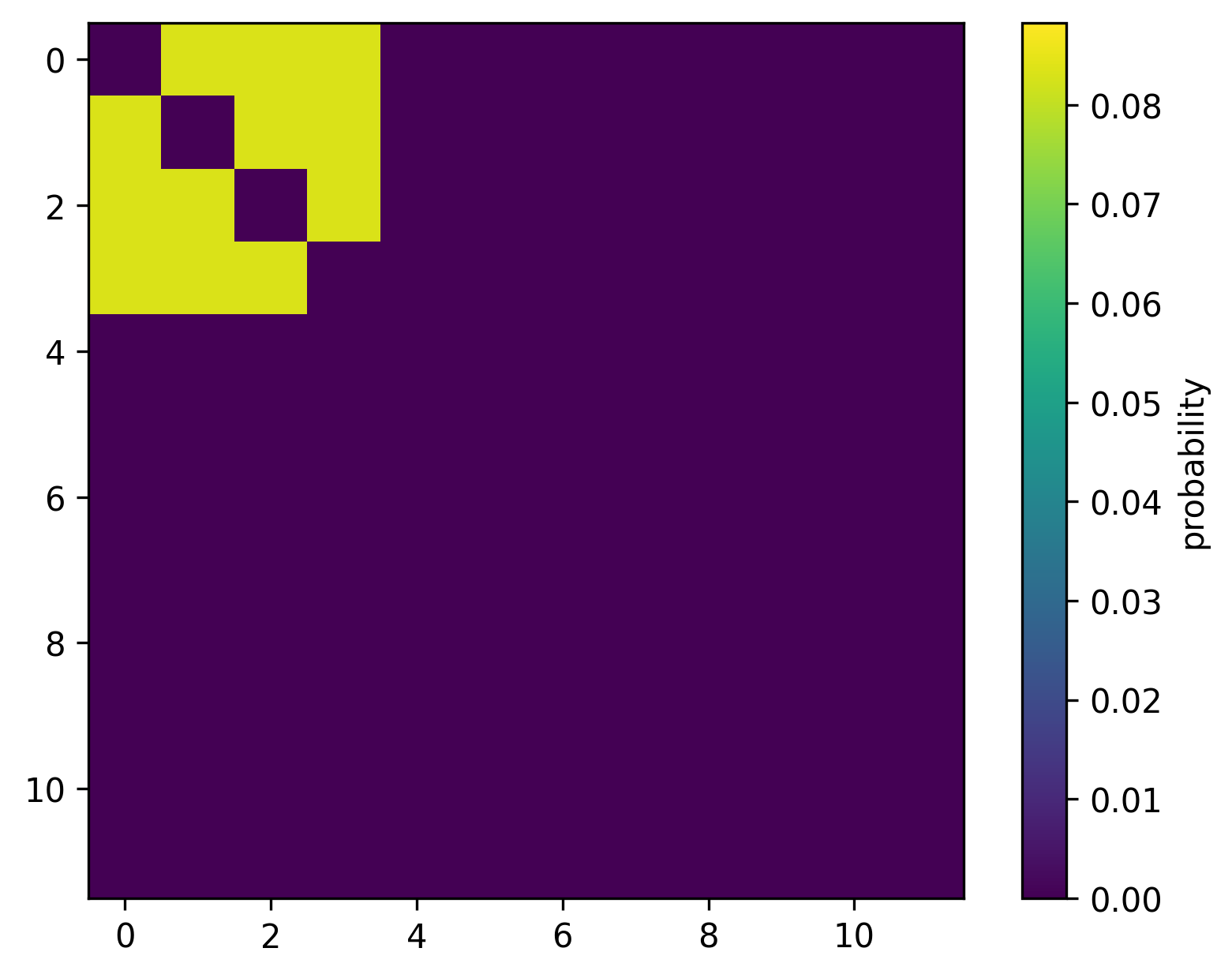}
    \subcaption{$t=0$}
  \end{minipage}
  \begin{minipage}{0.23\columnwidth}
    \centering
    \includegraphics[width=\columnwidth]{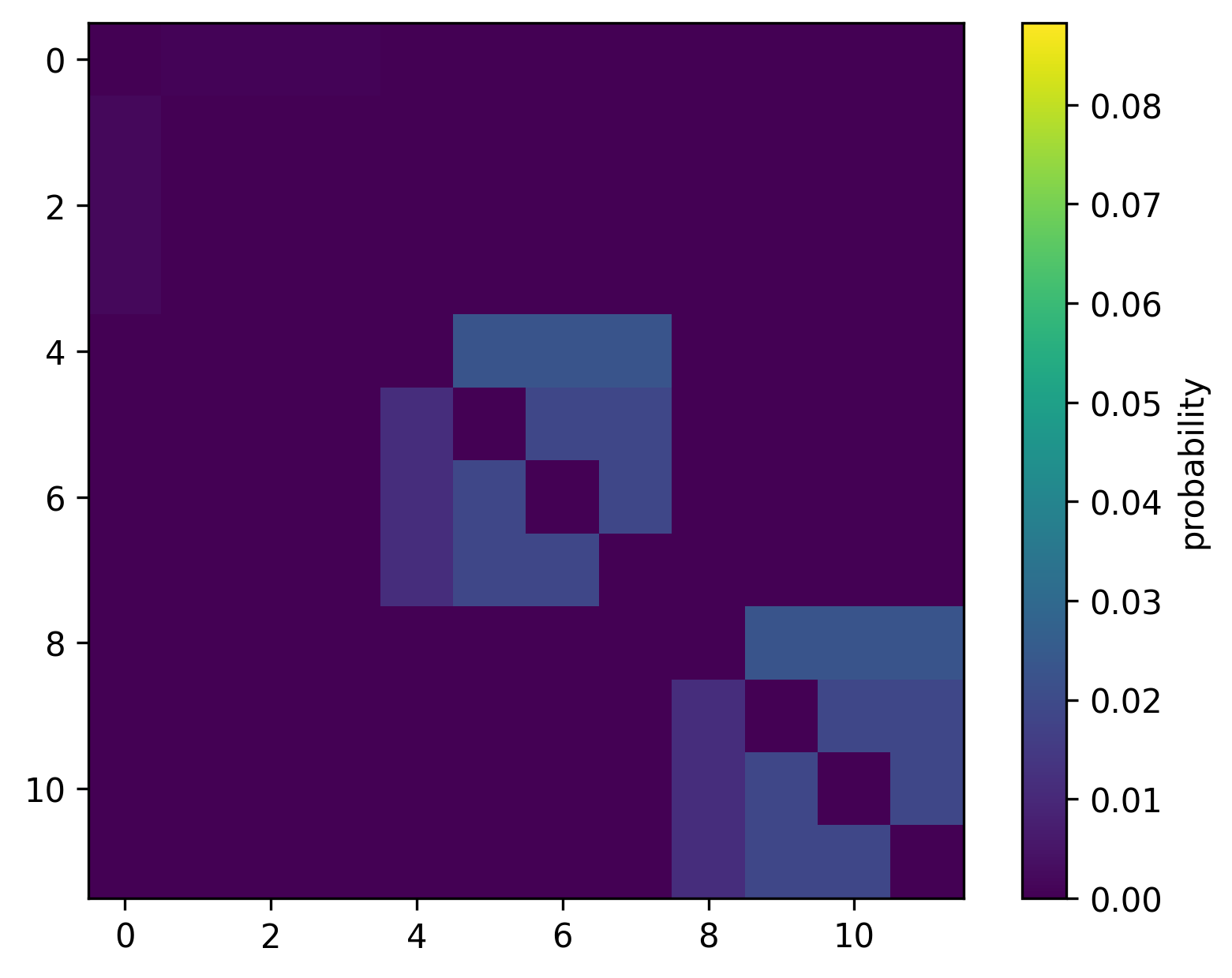}
    \subcaption{$t=13$}
  \end{minipage} 
  \begin{minipage}{0.23\columnwidth}
    \centering
    \includegraphics[width=\columnwidth]{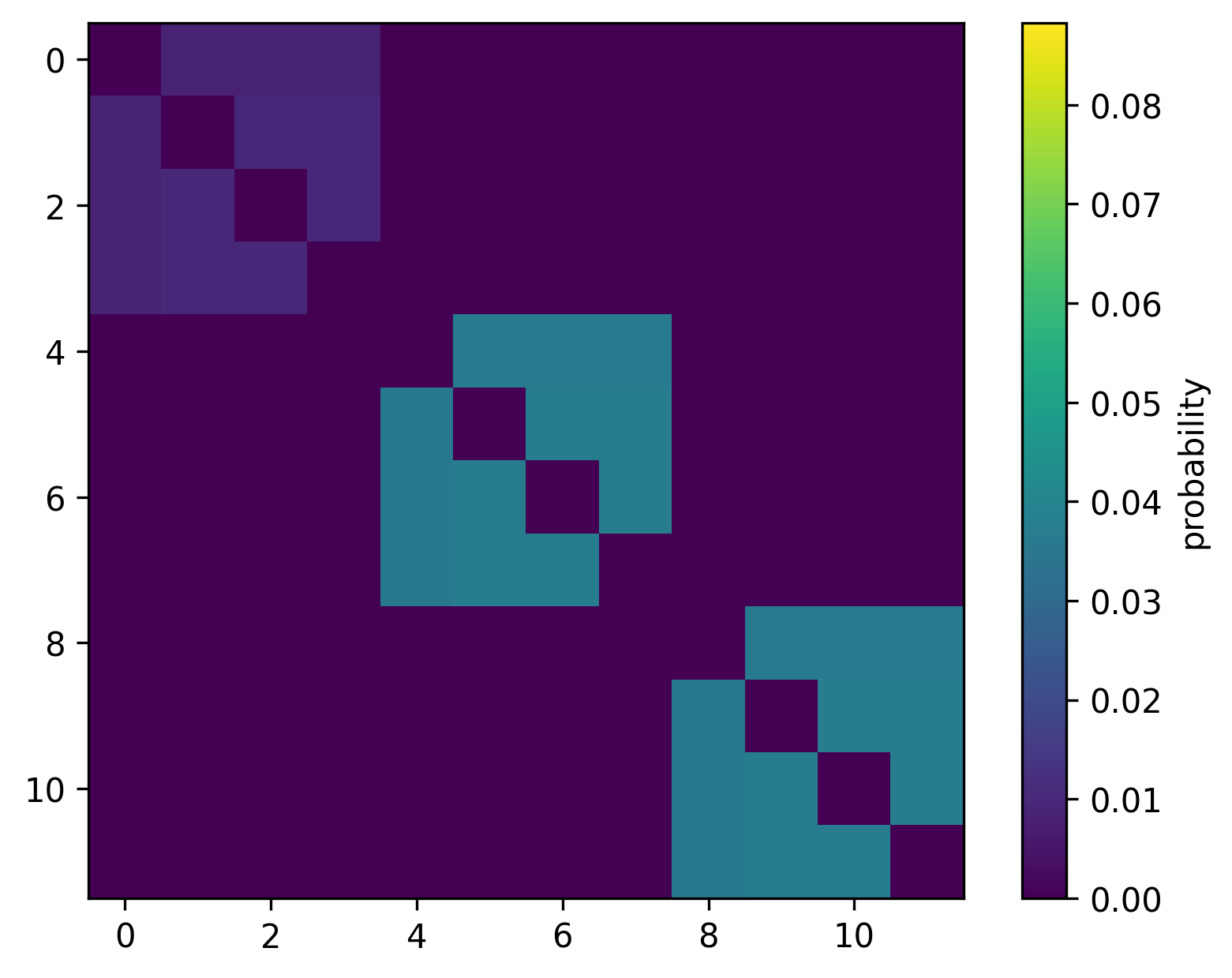}
    \subcaption{$t=20$}
  \end{minipage}
  \begin{minipage}{0.23\columnwidth}
    \centering
    \includegraphics[width=\columnwidth]{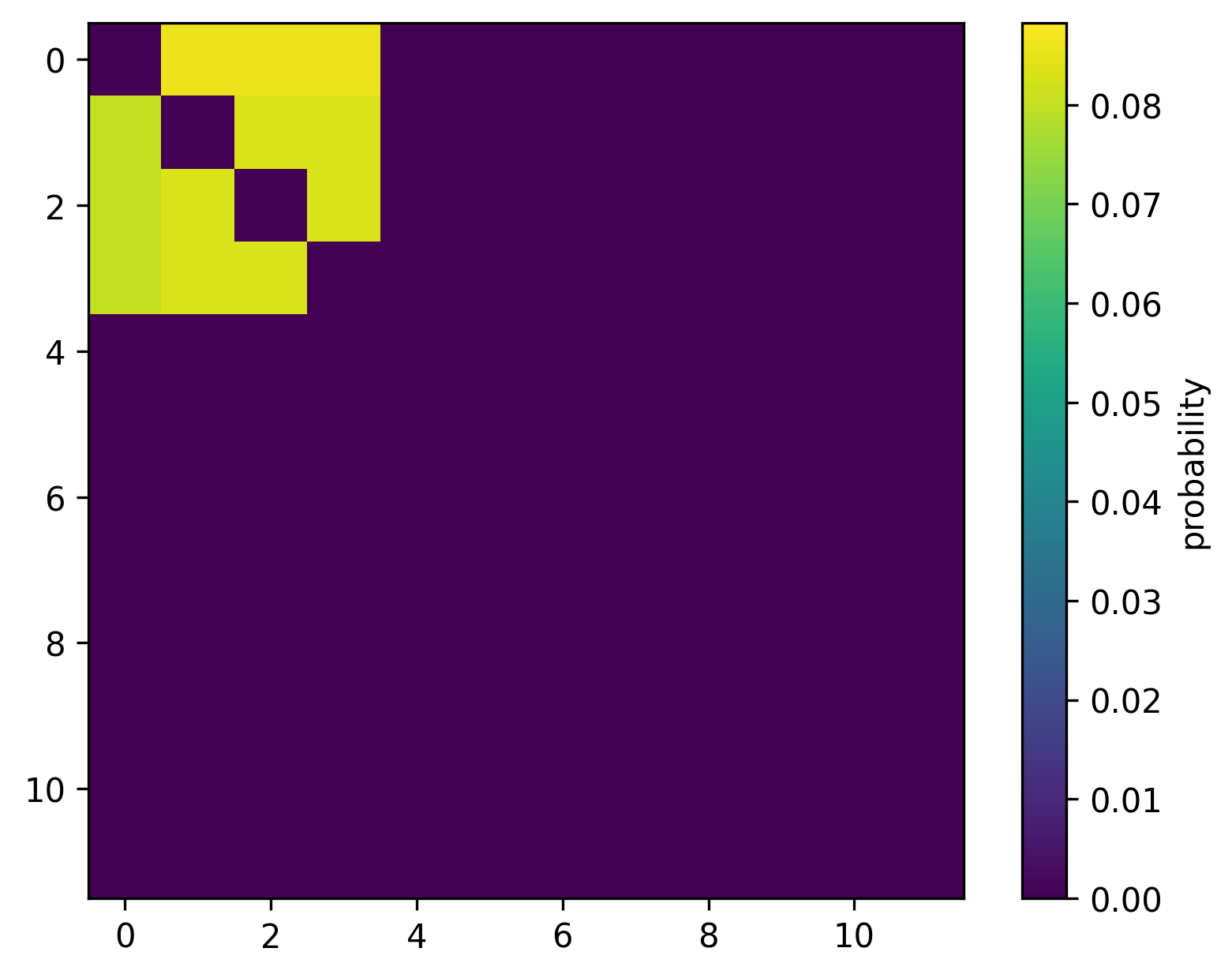}
    \subcaption{$t=40$}
  \end{minipage}
  \caption{
  Numerical simulation for the case \(\tilde G=C_3\) with \(G_1=G_2=G_3=K_4\) and \(\epsilon=0.05\). The upper-left panel compares the asymptotic finding probabilities predicted by Theorem~\ref{thm:prob_matrix} with the numerical simulation of the original quantum walk. The solid curves represent the theoretical values evaluated at the rescaled time \(\tau=t\sqrt{2\epsilon}\), while the dots represent the numerical results at each time step \(t\). The curves corresponding to \(G_2\) and \(G_3\) overlap and are therefore not visually distinguishable. The horizontal axis denotes the time step \(t\), and the vertical axis denotes the finding probability. The vertical dashed lines indicate the time steps corresponding to the heat maps shown below.
  The initial state is uniformly distributed on \(G_1\). The upper-right panel shows the original graph and the vertex labeling, where the red dashed edges represent the weak edges connecting the components. The lower panels show the probability distributions on the directed arcs at \(t=0,13,20,\) and \(40\). In each heat map, the \((i,j)\)-entry represents the probability on the directed arc from vertex \(i\) to vertex \(j\). The probabilities on the weak edges are omitted for clarity. Since the vertices in each component are labeled consecutively, the three diagonal blocks correspond to \(G_1\), \(G_2\), and \(G_3\), respectively.}
  \label{fig: cycle}
\end{figure}

\begin{figure}[htbp]
  \centering
  \begin{minipage}{0.55\columnwidth}
    \centering
    \includegraphics[width=\columnwidth]{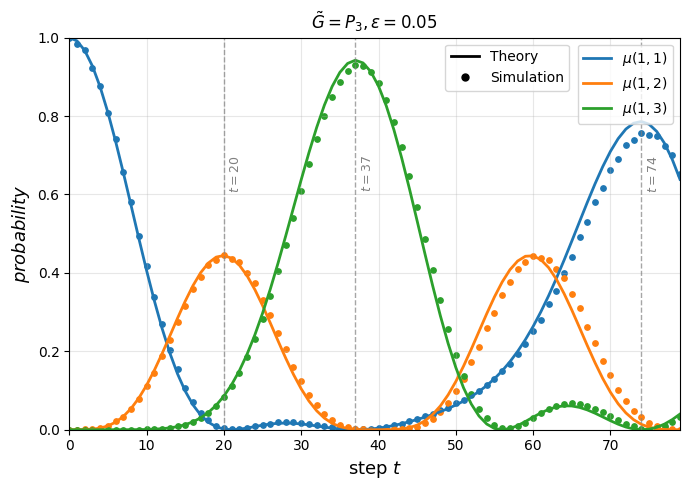}
  \end{minipage}
  \hspace{5mm}
  \begin{minipage}{0.35\columnwidth}
    \centering
    \includegraphics[width=\columnwidth]{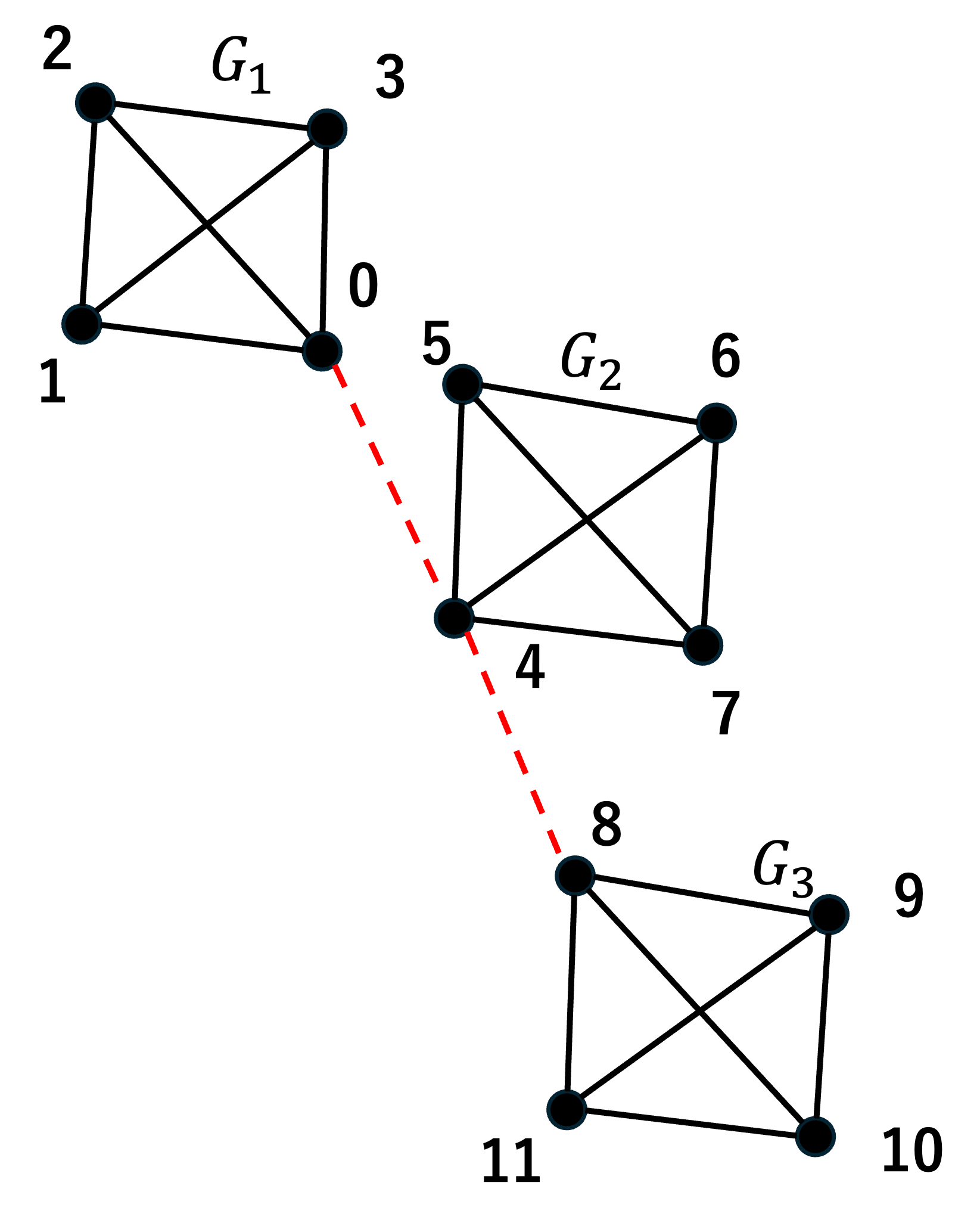}
  \end{minipage} \\

  \begin{minipage}{0.23\columnwidth}
    \centering
    \includegraphics[width=\columnwidth]{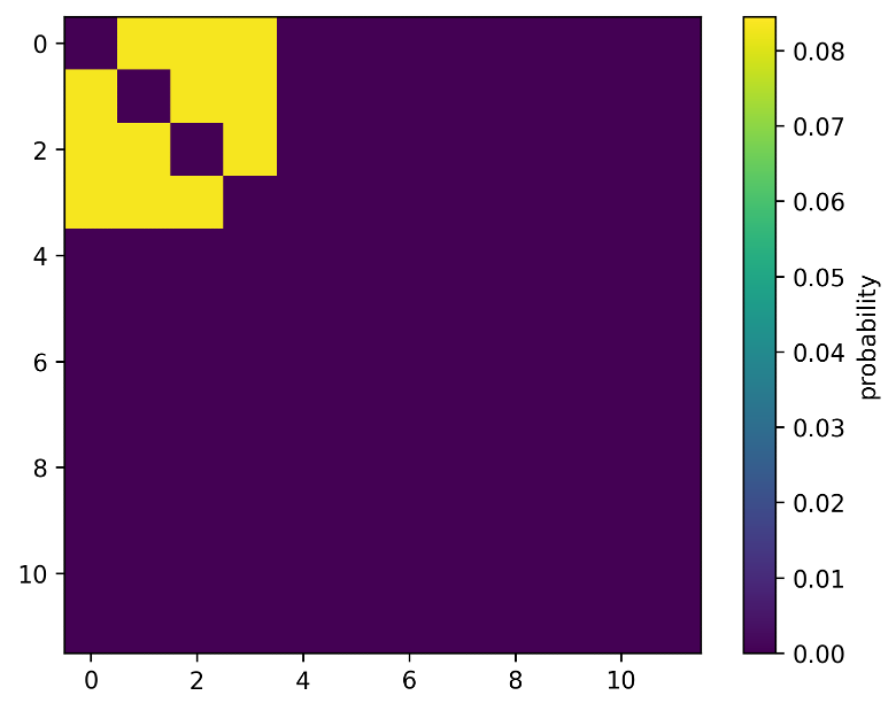}
    \subcaption{$t=0$}
  \end{minipage}
  \begin{minipage}{0.23\columnwidth}
    \centering
    \includegraphics[width=\columnwidth]{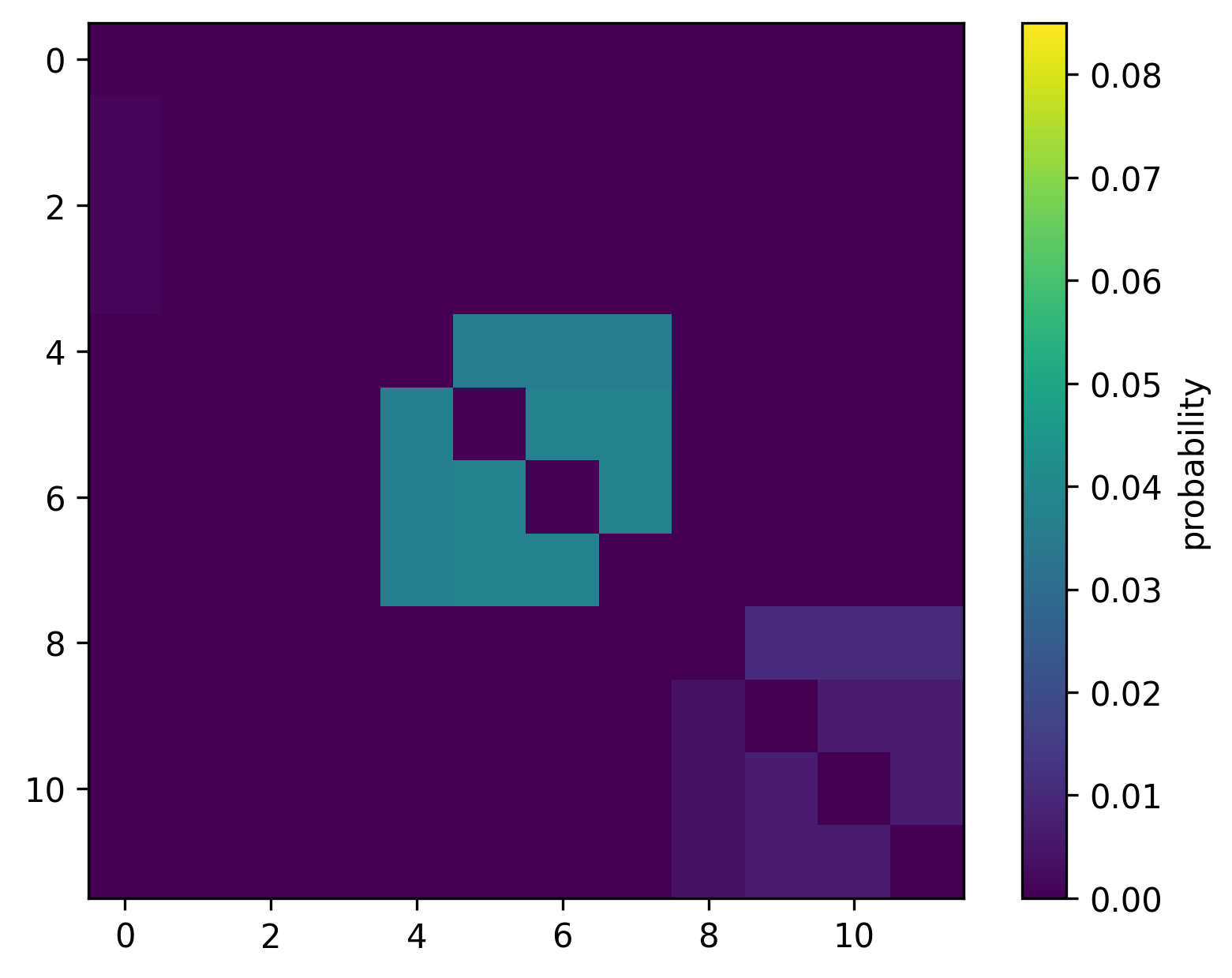}
    \subcaption{$t=20$}
  \end{minipage} 
  \begin{minipage}{0.23\columnwidth}
    \centering
    \includegraphics[width=\columnwidth]{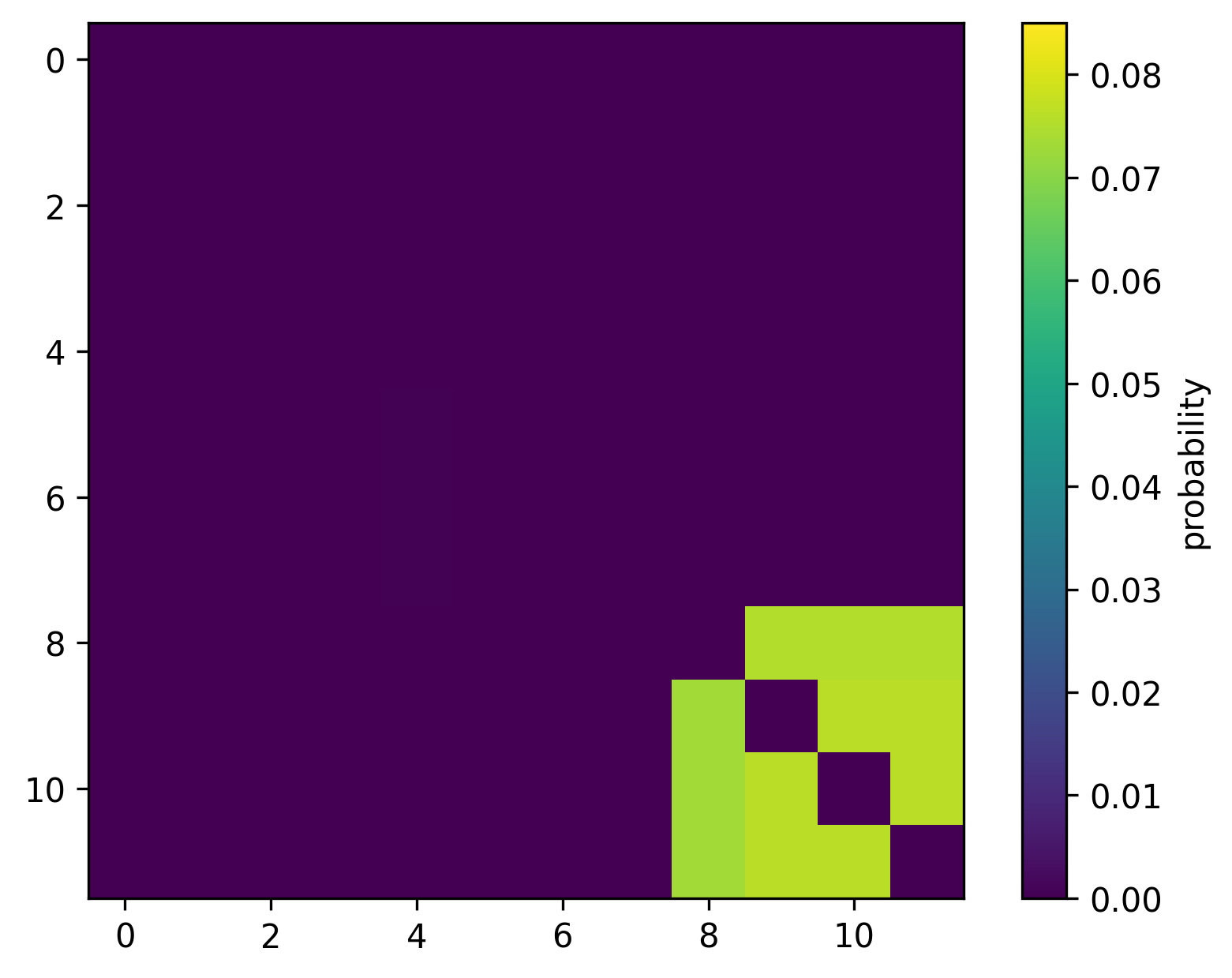}
    \subcaption{$t=37$}
  \end{minipage}
  \begin{minipage}{0.23\columnwidth}
    \centering
    \includegraphics[width=\columnwidth]{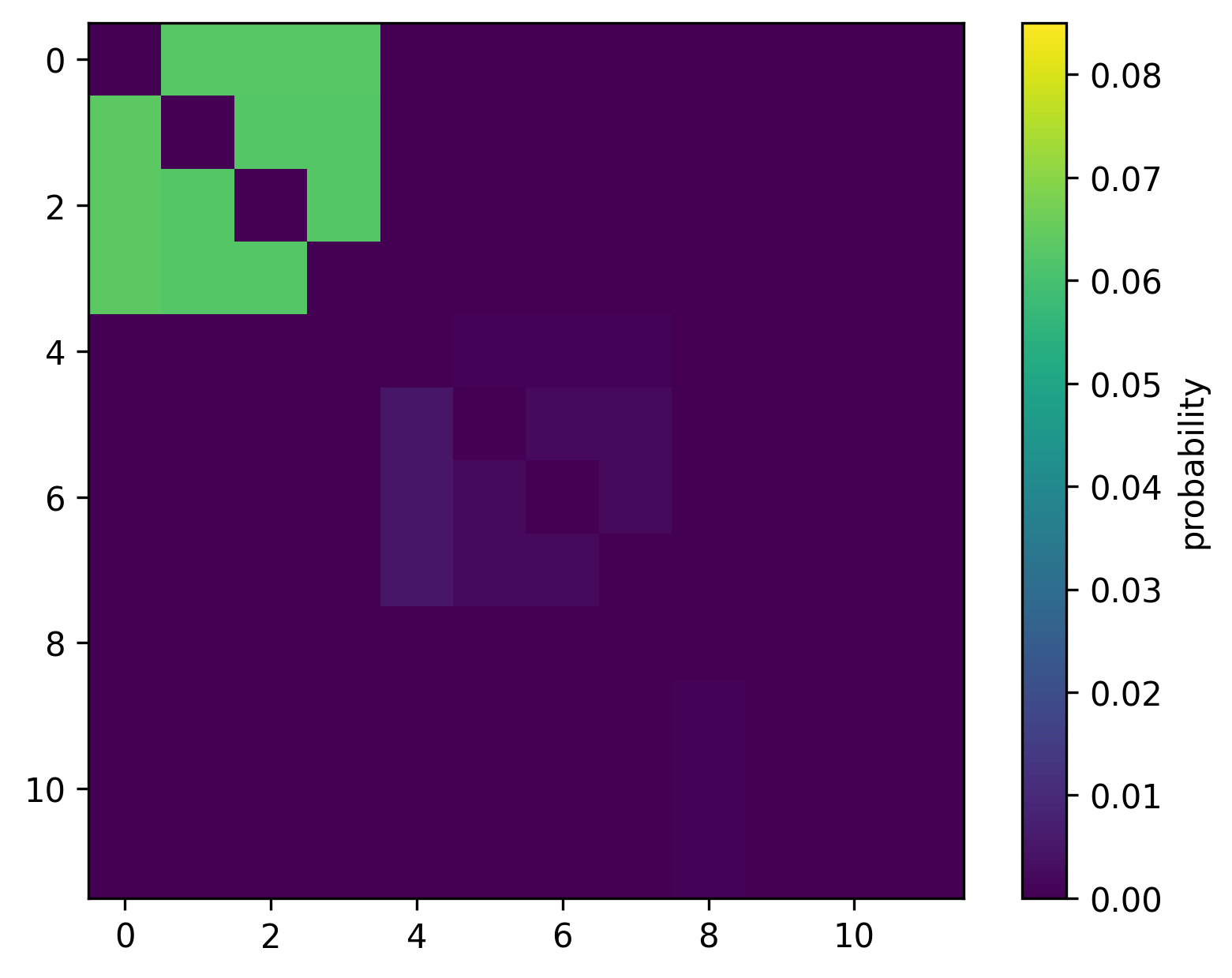}
    \subcaption{$t=74$}
  \end{minipage}
  \caption{
  Numerical simulation for the case \(\tilde G=P_3\) with \(G_1=G_2=G_3=K_4\) and \(\epsilon=0.05\). The upper-left panel compares the asymptotic finding probabilities predicted by Theorem~\ref{thm:prob_matrix} with the numerical simulation of the original quantum walk. The solid curves represent the theoretical values evaluated at the rescaled time \(\tau=t\sqrt{2\epsilon}\), while the dots represent the numerical results at each time step \(t\). The horizontal axis denotes the time step \(t\), and the vertical axis denotes the finding probability. The vertical dashed lines indicate the time steps corresponding to the heat maps shown below. The initial state is uniformly distributed on \(G_1\). The upper-right panel shows the original graph and the vertex labeling, where the red dashed edges represent the weak edges connecting the components. The lower panels show the probability distributions on the directed arcs at \(t=0,20,37,\) and \(74\). In each heat map, the \((i,j)\)-entry represents the probability on the directed arc from vertex \(i\) to vertex \(j\). The probabilities on the weak edges are omitted for clarity. The heat maps are interpreted in the same manner as in Fig.~\ref{fig: cycle}.
  }
  \label{fig: path}
\end{figure}


\clearpage

\subsection{Connection with a spring-mass system}

We now give a classical mechanical interpretation of the wave equation
obtained in Theorem~\ref{thm:prob_matrix}.
In particular, we show that the dynamics on the reduced graph has the same form as Newton's equation of motion for a spring-mass system.

Let $G=(X,A)$ be a simple connected graph with symmetric arc set $A$.
At each vertex $x\in X$, we place a point mass $m(x)>0$, and assign a
spring constant $k(a)>0$ to each arc $a\in A$. We assume
\begin{align*}
    k(a)=k(a^{-1})
\end{align*}
for every $a\in A$.

Let $u_\tau(x)$ denote the displacement of the mass at $x$ at time $\tau$.
Newton's equation of motion is given by
\begin{align}
\label{eq:motion}
    m(x)\frac{\partial^{2}}{\partial \tau^{2}}u_\tau(x)
    =
    -\sum_{a:t(a)=x}
    k(a)\bigl(u_\tau(x)-u_\tau(o(a))\bigr).
\end{align}

Let $D_k$ and $M_k$ be the weighted degree matrix and the weighted
adjacency matrix, respectively, defined by
\begin{align*}
    (D_kf)(x)
    &=
    \left(\sum_{a:t(a)=x}k(a)\right)f(x),\\
    (M_kf)(x)
    &=
    \sum_{a:t(a)=x}k(a)f(o(a)),
\end{align*}
for $f \in \mathbb{C}^{X}$.
We denote the corresponding weighted Laplacian by
\begin{align*}
    L_k=D_k-M_k.
\end{align*}
Furthermore, let $\mathcal{M}$ be the diagonal mass matrix defined by
\begin{align*}
    (\mathcal{M}f)(x)=m(x)f(x).
\end{align*}
Then Eq.~\eqref{eq:motion} can be written as
\begin{align*}
    \mathcal{M}\frac{\partial^{2}}{\partial t^{2}}\boldsymbol{u}_\tau
    =
    -L_k\boldsymbol{u}_\tau.
\end{align*}
Here, we set
\begin{align*}
    \boldsymbol{v}_\tau
    =
    \mathcal{M}^{1/2}\boldsymbol{u}_\tau.
\end{align*}
Then, we obtain
\begin{align}
\label{eq:spring_laplacian}
    \frac{\partial^{2}}{\partial \tau^{2}}\boldsymbol{v}_\tau
    =
    -\mathcal{M}^{-1/2}L_k\mathcal{M}^{-1/2}
    \boldsymbol{v}_\tau.
\end{align}
Thus, the motion of the spring-mass system is governed by the symmetric
operator $\mathcal{M}^{-1/2}L_k\mathcal{M}^{-1/2}$.

We now apply this formulation to the reduced graph $\tilde{G}$.
For each $V\in\tilde{X}$, assign the mass
\begin{align*}
    m(V)=|A_V|,
\end{align*}
and for $V,W\in\tilde{X}$, assign the spring constant
\begin{align*}
    k(V,W)=\ell(V,W).
\end{align*}
Then the mass matrix satisfies
\begin{align*}
    \mathcal{M}^{-1}=R,
\end{align*}
while the weighted Laplacian of the spring-mass system is precisely the
Laplacian $L$ on the reduced graph. Hence,
\begin{align*}
    \mathcal{M}^{-1/2}L\mathcal{M}^{-1/2}
    =
    R^{1/2}LR^{1/2}
    =
    L_{\mathrm{sym}}.
\end{align*}

Therefore, Eq.~\eqref{eq:spring_laplacian} becomes
\begin{align*}
    \frac{\partial^{2}}{\partial \tau^{2}}\boldsymbol{v}_\tau
    =
    -L_{\mathrm{sym}}\boldsymbol{v}_\tau,
\end{align*}
which has the same form as the wave equation of 
Eq.~\eqref{eq: wave_reduced} in Theorem~\ref{thm:prob_matrix}. 
Then, if the quantum walk starts from the connected component $X_j$, the modified displacement of the mass $\boldsymbol{v}_\tau$ at time $\tau$ in the spring mass system on $\tilde{G}$ with the initial condition $\boldsymbol{u}_0(x)=\delta_{x,j}$ and $(\partial \boldsymbol{u}_\tau/\partial \tau)|_{\tau=0}=0$ coincides with the amplitude $\Phi_\tau(j,\cdot)$ of Eq.~\eqref{eq: wave_reduced} in Theorem~\ref{thm:prob_matrix}. 
This concludes that the wave equation describing the asymptotic transport between the components in our quantum walk system can be interpreted as the equation of motion of a classical spring-mass system on the reduced graph. The correspondence between the two models is summarized in Table~\ref{tab:spring_mass}.

\begin{table}[htbp]
\centering
\renewcommand{\arraystretch}{1.5}
\begin{tabular}{ll|ll}
\hline
\multicolumn{2}{c|}{{\bf Quantum walk system}} 
&\multicolumn{2}{c}{{\bf Spring-mass system}} \\
\hline
$V\in \tilde{X}$ & (connected component)  & $V\in \tilde{X}$ & (point mass) \\
$|A_V|$ & ($\sharp$ arcs in $V$)  & $m(V)$ &(mass at point mass $V$) \\
$\ell(V,W)$ & ($\sharp$ edges between $V$ and $W$)  & $k(V,W)$ & (spring constant between $V$ and $W$)  \\
\hline
\end{tabular}
\caption{Correspondence between the quantum walk system and the spring-mass system.}
\label{tab:spring_mass}
\end{table}

\section{Proof of main theorem (Theorem~\ref{thm:prob_matrix})}
\label{subsec_proof_many}
\subsection{Spectral analysis of the perturbed transition matrix}
 
 Let $W(\epsilon): \mathbb{C}^{X} \to \mathbb{C}^{X}$ be the transition matrix on $G$ defined by
\begin{align}
    \label{eq: RW}
    (W(\epsilon)f)(x) = \sum_{a: o(a)=x} p_{\epsilon}(a)f(t(a))
\end{align}
for $f \in \mathbb{C}^{X}$ and $x \in X$. 
By definition, the transition probability $p_{\epsilon}(a)$ can be expanded in terms of $\epsilon$ as follows:
\begin{align*}
    p_{\epsilon}(a) =
    \begin{dcases}
        \frac{1}{\deg(o(a))} - \frac{\deg_{B}(o(a))}{\deg(o(a))^{2}}\epsilon + O(\epsilon^{2}) &: o(a) \in \delta X \text{ and } a \notin B, \\
        \frac{1}{\deg(o(a))} &: o(a) \notin \delta X \text{ and } a \notin B, \\
        \frac{1}{\deg(o(a))}\epsilon + O(\epsilon^{2}) &: a \in B.
    \end{dcases}
\end{align*}
Here, $\delta X \subset X$ denote the set of boundary vertices given by
\begin{align*}
    \delta X=\{x \in X : \deg_B(x) > 0\}.
\end{align*}
Consequently, the transition matrix $W(\epsilon)$ can be expanded as
\begin{align*}
    (W(\epsilon)f)(x) = (W^{(0)}f)(x) + \epsilon(W^{(1)}f)(x) + O(\epsilon^{2}),
\end{align*}
where the unperturbed term $W^{(0)}$ and the first-order perturbation term $W^{(1)}$ are explicitly expressed as
\begin{align*}
    (W^{(0)}f)(x) &= \sum_{\substack{a:o(a)=x \\ a \notin B}} \frac{1}{\deg(o(a))} f(t(a)), \\
    (W^{(1)}f)(x) &= \sum_{\substack{a:o(a)=x \\ a \in B}} \frac{1}{\deg(o(a))} f(t(a)) - \sum_{\substack{a:o(a)=x \\ o(a) \in \delta X, a \notin B}} \frac{\deg_{B}(o(a))}{\deg(o(a))^{2}} f(t(a)).
\end{align*}

Note that $W^{(0)}$ is the transition matrix of the simple random walk on $G\setminus B$, that is, $W^{(0)}$ can be decomposed into the direct sum
\begin{align}
    \label{eq: W0}
    W^{(0)}=\bigoplus_{j=1}^{n}W_{j}
\end{align}
where $W_{j}$ is the transition matrix of the simple random walk on $G_{j}$ and $n$ is the number of connected components for $\epsilon=0$. From the property of the transition matrices,
$W_j$ has the simple eigenvalue $1$ and its eigenprojection is given by
\begin{align}
    \label{eq: eigenprojection_Wj}
    \Pi_{j}=K_{j}M_{j}^{*}
\end{align}
where the operators $K_{j}, M_{j}: \mathbb{C}\rightarrow\mathbb{C}^{X_{j}}$ correspond to right and left eigenvectors of eigenvalue $1$, respectively, which are defined by
\begin{align*}
    (K_{j}c)(x)=c, \quad (M_{j}c)(x)=m_{j}(x)\cdot c
\end{align*}
for $c \in \mathbb{C}$ and $ x \in X_{j}$.
Here, $m_{j}$ is the stationary distribution on the connected component $G_{j}$ defined by
\begin{align*}
    m_{j}(x)=\frac{\deg(x)}{|A_{j}|}
\end{align*}
for $x \in X_{j}$. Then, by Eqs.~\eqref{eq: W0} and \eqref{eq: eigenprojection_Wj}, the eigenprojection $\Pi$ onto the eigenspace of $W^{(0)}$ corresponding to the eigenvalue $1$ is given by
\begin{align}
    \label{eq: pi}
    \Pi=K_{\pi}M_{\pi}^{*}
\end{align}
where the operators $K_{\pi}, M_{\pi}: \mathbb{C}^{\tilde{X}}\rightarrow\mathbb{C}^{X}$ are defined by
\begin{align*}
    K_{\pi}=\bigoplus_{j=1}^{n}K_{j}, \quad M_{\pi}=\bigoplus_{j=1}^{n}M_{j}.
\end{align*}
By applying the definitions of the direct sums, we can explicitly write down the actions of $K_{\pi}$ and $M_{\pi}$. For convenience, we denote the stationary distribution on the component containing $x$ as $m_{\pi(x)}$, that  is, $m_{\pi(x)} = m_{j}$ if $x \in X_{j}$. 
Then, we immediately obtain
\begin{align*}
    (K_{\pi}\tilde{f})(x) &= \tilde{f}(\pi(x)), \qquad \qquad \quad (K^{*}_{\pi}f)(V)=\sum_{x: \pi(x)=V} f(x), \\
    (M_{\pi}\tilde{f})(x) &= m_{\pi(x)}(x)\tilde{f}(\pi(x)),\quad (M^{*}_{\pi}f)(V)= \sum_{x: \pi(x)=V} \overline{m_{\pi(x)}(x)}f(x)
\end{align*}
for any $\tilde{f} \in \mathbb{C}^{\tilde{X}}$, $f \in \mathbb{C}^{X}$, $x \in X$, and $V \in \tilde{X}$.
Note that $M_{\pi}^{*}K_{\pi}$ is the identity operator in $\mathbb{C}^{\tilde{X}}$.

We should remark that the eigenvalue $1$ of $W(0)$ has the multiplicity $n$ (number of connected components), but obviously it is semi-simple.
According to perturbation theory \cite{kato1982}, the eigenvalue \(1\) of \(W(0)\) splits under the perturbation, and we can apply the reduction process as follows. The first-order splitting is determined by the eigeninformation of the reduced operator
\begin{align}
\label{eq: W1hat}
\hat{W}^{(1)}=\Pi W^{(1)}\Pi.
\end{align}
If $\alpha \in \mathrm{Spec}(\hat{W}^{(1)}|_{\mathrm{Ran}(\Pi)})$ has multiplicity $m_{\alpha}$, and is semi-simple, then 
the perturbed eigenvalue of $W(\epsilon)$ splitting from the non-perturbed eigenvalue $1$ can be expanded by
\begin{align*}
\nu_{\alpha,s}(\epsilon)=1+\epsilon\alpha+O(\epsilon^{2}),\quad s=1,...,m_{\alpha}.
\end{align*}
Let \(P_\alpha\) be the total eigenprojection of \(\hat W^{(1)}|_{\operatorname{Ran}(\Pi)}\) associated with \(\alpha\), and let \(P_\alpha^{(\epsilon)}\) be the total eigenprojection of \(W(\epsilon)\) associated with $\nu_{\alpha,s}(\epsilon)$. Then, it follows that
\begin{align*}
    P_\alpha^{(\epsilon)}=P_\alpha+O(\epsilon).
\end{align*}
The expansions for terms of the second order and higher also can be done in a similar fashion to the above as long as the corresponding non-perturbed matrix to ${\hat{W}}^{(1)}$ at each time step is semi-simple. We note that since $W(\epsilon)$ is similar to a self-adjoint matrix (see Eq.~(\ref{eq:similar})), every perturbed eigenvalue is semi-simple for any $\epsilon\geq 0$. Thus, the reduction process comes to end until finite steps and the resulting total eigenprojection at the finial step provides the completely expansion of the eigenprojection itself without nilpotent. Therefore, every perturbed eigenprojection of the perturbed eigenvalue splitting from the non-perturbed eigenvalue $1$ is analytic at $1$. 
Thus, we can choose the corresponding eigenvectors $g_{\alpha, s}^{(\epsilon)} \in \mathrm{Ker}(W(\epsilon)-\nu_{\alpha,s}(\epsilon) I)$, such that
\begin{align*}
    g_{\alpha, s}^{(\epsilon)}=g_{\alpha, s}^{(0)}+O(\epsilon).
\end{align*}
where $g_{\alpha, s}^{(0)} \in \mathrm{Ker}(\hat{W}^{(1)}-\alpha I)$. 

Thus, the expansion until the first order with respect to $\epsilon$ of the eigenvalues and its eigenprojections of $W(\epsilon)$ that converge to $1$ for $\epsilon\to 0$ is reduced to the spectral analysis of $\hat{W}^{(1)}$. The following proposition shows that this operator can be explicitly expressed in terms of the Laplacian on the reduced graph $\tilde{G}$.

\begin{proposition}
    \label{prop: laplacian}
    Let $L$ and $R$ be defined by Eqs.~(\ref{eq: laplacian}) and (\ref{eq: def_R}), respectively. Then, it follows that
    \begin{align*}
        \hat{W}^{(1)}=-K_{\pi}RLM_{\pi}^{*}.
    \end{align*}
\end{proposition}
\begin{proof}
We evaluate the action of the operator $M^{*}_{\pi}W^{(1)}K_{\pi}$.
For any $\tilde{f} \in \mathbb{C}^{\tilde{X}}$ and $ V \in \tilde{X}$, we obtain
\begin{align*}
&(M_{\pi}^{*}W^{(1)}K_{\pi}\tilde{f})(V) \\
&=
\sum_{x:\pi(x)=V}
m_V(x)(W^{(1)}K_{\pi}\tilde{f})(x)\\
&=
\frac{1}{|A_V|}
\left(
\sum_{x:\pi(x)=V}
\sum_{\substack{a:o(a)=x\\a\in B}}
\tilde{f}(\pi(t(a)))
-
\sum_{x:\pi(x)=V}
\deg_B(x)\tilde{f}(V)
\right)\\
&=
\frac{1}{|A_V|}
\left(
\sum_{W\in\tilde X}\ell(V,W)\tilde{f}(W)
-
\deg_{\tilde G}(V)\tilde{f}(V)
\right)\\
&=
-(RL\tilde{f})(V).
\end{align*}
Therefore, by Eqs. (\ref{eq: pi}), (\ref{eq: W1hat}) and the above formula, we have
\begin{align*}
\hat{W}^{(1)}
&=\Pi W^{(1)}\Pi\\
&=K_{\pi}M_{\pi}^{*}W^{(1)}K_{\pi}M_{\pi}^{*}\\
&=-K_{\pi}RLM_{\pi}^{*}.
\end{align*}
This completes the proof.
Hence, we obtain the desired result.
\end{proof}
As implied by Proposition~\ref{prop: laplacian}, the spectrum of $\hat{W}^{(1)}\in \mathbb{C}^{V\times V}$ is completely characterized by that of the operator $RL\in \mathbb{C}^{\tilde{X}\times\tilde{X}}$. The following proposition explicitly constructs the eigenvectors of $\hat{W}^{(1)}$ by lifting those of $RL$ using $K_\pi:\mathbb{C}^{\tilde{X}}\to \mathbb{C}^V$.
\begin{lemma}
    \label{lem: e-val}
    Let $\tilde{g}_{\lambda} \in \mathrm{Ker}(RL-\lambda I)$ be an eigenvector on $\tilde{G}$. Then, we have
    \begin{align*}
        K_{\pi}\tilde{g}_{\lambda} \in \mathrm{Ker}(\hat{W}^{(1)}+\lambda I).
    \end{align*}
    Additionally, the eigenvector $K_{\pi}\tilde{g}_{\lambda}$ belongs to the eigenspace $\mathrm{Range}(\Pi)$.
\end{lemma}
\begin{proof}
    Since $RL\tilde{g}_{\lambda}=\lambda\tilde{g}_{\lambda}$ and $M_{\pi}^{*}K_{\pi}=I$ hold, we have
    \begin{align*}
        -K_{\pi}RLM_{\pi}^{*}K_{\pi}\tilde{g}_{\lambda}&=-\lambda K_{\pi}\tilde{g}_{\lambda}.
    \end{align*}
    Then, from Proposition~\ref{prop: laplacian}, we immediately see
    \begin{align*}
        \hat{W}^{(1)}K_{\pi}\tilde{g}_{\lambda}&=-\lambda K_{\pi}\tilde{g}_{\lambda}.
    \end{align*}
    Thus, $K_{\pi}\tilde{g}_{\lambda} \in \mathrm{Ker}(\hat{W}^{(1)}+\lambda I)$ holds.
    Moreover, by the definition of $\Pi = K_{\pi}M_{\pi}^{*}$, it follows that
    \begin{align*}
        \Pi (K_{\pi}\tilde{g}_{\lambda}) = K_{\pi}(M_{\pi}^{*}K_{\pi})\tilde{g}_{\lambda} = K_{\pi}\tilde{g}_{\lambda},
    \end{align*}
    which implies $K_{\pi}\tilde{g}_{\lambda} \in \mathrm{Range}(\Pi)$. This completes the proof.
\end{proof}
From Proposition~\ref{prop: laplacian} and Lemma~\ref{lem: e-val}, the spectral analysis of $W(\epsilon)$ around $1$ is reduced to that of $RL$. 
However, the matrix $RL$ is not symmetric in general, making the problem more complicated.
To address this, we rewrite Lemma~\ref{lem: e-val} using the symmetric matrix $L_{\mathrm{sym}}$ defined in Eq.~(\ref{eq: sym_laplacian}).

\begin{lemma}
    \label{lem: e-val_sym}
    Let $\tilde{f}_{\lambda} \in \mathrm{Ker}(L_{\mathrm{sym}}-\lambda I)$ be an eigenvector on $\tilde{G}$. Then, it follows that
    \begin{align*}
        K_{\pi}R^{1/2}\tilde{f}_{\lambda} \in \mathrm{Ker}(\hat{W}^{(1)}+\lambda I).
    \end{align*}
    Additionally, the eigenvector $K_{\pi}R^{1/2}\tilde{f}_{\lambda}$ belongs to the eigenspace $\mathrm{Range}(\Pi)$.
\end{lemma}
\begin{proof}
    Since $L_{\mathrm{sym}}\tilde{f}_{\lambda} = \lambda\tilde{f}_{\lambda}$ and $L_{\mathrm{sym}} = R^{1/2}LR^{1/2}$, we immediately see
    \begin{align*}
        RL(R^{1/2}\tilde{f}_{\lambda}) = \lambda(R^{1/2}\tilde{f}_{\lambda}).
    \end{align*}
    This implies $R^{1/2}\tilde{f}_{\lambda} \in \mathrm{Ker}(RL-\lambda I)$. 
    Applying Proposition~\ref{lem: e-val} with $\tilde{g}_{\lambda} = R^{1/2}\tilde{f}_{\lambda}$, we immediately obtain the desired result.
\end{proof}
Lemma~\ref{lem: e-val_sym} shows that the spectral analysis of $W(\epsilon)$ is reduced to that of $L_{\mathrm{sym}}$.
By combining the previous results, we establish the following key lemma.

\begin{lemma}
\label{lem: eigenvalue_expansion}
Let $W(\epsilon)$ be the transition matrix on $G$ defined in Eq.~\eqref{eq: RW}.
Let $\lambda\in\operatorname{Spec}(L_{\mathrm{sym}})$ have multiplicity
$m_\lambda$. Then $m_\lambda$ eigenvalues of $W(\epsilon)$, counted with
multiplicity, split from $1$ and satisfy
\begin{align*}
    \nu_{\lambda,s}(\epsilon)
    =
    1-\epsilon\lambda+O(\epsilon^2),
    \qquad
    s=1,\ldots,m_\lambda,
\end{align*}
as $\epsilon\to0$.

Moreover, there exist orthonormal eigenvectors $\tilde{f}_{\lambda, s} \in \mathrm{Ker}(L_{\mathrm{sym}}-\lambda I), (s=1,...,m_{\lambda})$ such that the corresponding eigenvectors
$g_{\lambda,s}^{(\epsilon)}$ of $W(\epsilon)$ can be chosen to satisfy
\begin{align*}
    g_{\lambda,s}^{(\epsilon)}
    =
    K_\pi R^{1/2}\tilde f_{\lambda,s}
    +O(\epsilon),
    \qquad
    s=1,\ldots,m_\lambda.
\end{align*}
\end{lemma}

\subsection{Applying the spectral mapping theorem}

In this subsection, we consider the eigenvalues and eigenvectors of $U(\epsilon)$ by using the spectral mapping theorem for quantum walks \cite{higuchi2014}.
We define the diagonal operator $D_{\epsilon}^{1/2}: \mathbb{C}^{X}\rightarrow\mathbb{C}^{X}$ by
\begin{align*}
(D_{\epsilon}^{1/2}f)(x)
=
\sqrt{m_{\epsilon}(x)}f(x)
\end{align*}
for any $f\in\mathbb{C}^{X}$ and $x\in X$.
Using this operator, we define the symmetrized transition matrix associated with $W(\epsilon)$ by
\begin{align}\label{eq:similar}
W_{\mathrm{sym}}(\epsilon)
=
D_{\epsilon}^{1/2}W(\epsilon)D_{\epsilon}^{-1/2}.
\end{align}
Since $W_{\mathrm{sym}}(\epsilon)$ and $W(\epsilon)$ are similar, they have the same spectrum.

Let \(\lambda\in\operatorname{Spec}(L_{\mathrm{sym}})\) have multiplicity \(m_\lambda\). We denote by \(\nu_{\lambda,s}(\epsilon)\), \(s=1,\ldots,m_\lambda\), the eigenvalues of \(W(\epsilon)\) corresponding to \(\lambda\) as described in Lemma~\ref{lem: eigenvalue_expansion}, and by \(g_{\lambda,s}^{(\epsilon)}\) their corresponding eigenvectors.
Then, $\nu_{\lambda,s}(\epsilon)$ is also an eigenvalue of $W_{\mathrm{sym}}(\epsilon)$.
The corresponding eigenvector is given by
\begin{align}
\label{eq: eigenvector_wsym}
f_{\lambda,s}^{(\epsilon)}
=
D_{\epsilon}^{1/2}g_{\lambda,s}^{(\epsilon)}.
\end{align}
After normalization, we use the same notation $f_{\lambda,s}^{(\epsilon)}$ for the normalized eigenvector.
Then, we apply the spectral mapping theorem of quantum walks for our setting.
\begin{proposition}[\cite{higuchi2014}]
\label{prop: SMT}
Let $\lambda\in\mathrm{Spec}(L_{\mathrm{sym}})$, and let
$f_{\lambda,s}^{(\epsilon)}
\in
\mathrm{Ker}
\left(
W_{\mathrm{sym}}(\epsilon)
-
\nu_{\lambda,s}(\epsilon)I
\right), s=1,...,m_{\lambda}$
be normalized eigenvectors.
We set
\begin{align*}
\cos\theta_{\lambda,s}(\epsilon)
=
\nu_{\lambda,s}(\epsilon),
\qquad
\theta_{\lambda,s}(\epsilon)\in[0,\pi].
\end{align*}
For $\lambda\neq0$, the corresponding eigenvectors
$\eta_{\sigma,\lambda,s}^{(\epsilon)}
\in
\mathrm{Ker}
\left(
U(\epsilon)
-
e^{\sigma i\theta_{\lambda,s}(\epsilon)}I
\right)$,
$\sigma\in\{\pm1\}$, are given by
\begin{align*}
\eta_{\sigma,\lambda,s}^{(\epsilon)}
=
\frac{1}{\sqrt{2}|\sin\theta_{\lambda,s}(\epsilon)|}
\left(
I-e^{\sigma i\theta_{\lambda,s}(\epsilon)}S
\right)
d_{\epsilon}^{*}f_{\lambda,s}^{(\epsilon)}.
\end{align*}
For $\lambda=0$, the corresponding eigenvector
$\eta_{1}^{(\epsilon)}
\in
\mathrm{Ker}(U(\epsilon)-I)$
is given by
\begin{align*}
\eta_{1}^{(\epsilon)}
=
d_{\epsilon}^{*}f_{0}^{(\epsilon)}.
\end{align*}
More explicitly, for each $a\in A$, it follows that
\begin{align*}
\eta_{\sigma,\lambda,s}^{(\epsilon)}(a)
&=
\frac{1}{\sqrt{2}|\sin\theta_{\lambda,s}(\epsilon)|}
\left(
\sqrt{p_{\epsilon}(a^{-1})}
f_{\lambda,s}^{(\epsilon)}(t(a))
-
e^{\sigma i\theta_{\lambda,s}(\epsilon)}
\sqrt{p_{\epsilon}(a)}
f_{\lambda,s}^{(\epsilon)}(o(a))
\right), \\
\eta_{1}^{(\epsilon)}(a)
&=
\sqrt{p_{\epsilon}(a^{-1})}
f_{0}^{(\epsilon)}(t(a)).
\end{align*}

\end{proposition}
The subspace spanned by all the eigenvectors in the above proposition, $\mathrm{span}\{\eta^{(\epsilon)}_{\sigma,\lambda,s}\;|\; \lambda\in \mathrm{Spec}(L_{sym}), \sigma\in\{\pm\},s=1,\dots,m_\lambda \}$, is equivalent to the invariant subspace of $U(\epsilon)$,  $\mathcal{L}:=\mathrm{Ran}(d^*)+\mathrm{Ran}(Sd^*)$, which is called the inherited eigenspace from $W_{\mathrm{sym}}(\epsilon)$. In this paper, the inherited eigenspace lifted up from the perturbed eigenvalues around $1\in \mathrm{spec}(W_{sym}(\epsilon))$ is further inherited from $L_{sym}$ (see Lemma~\ref{lem: e-val_sym}). Since the birth eigenspace $\mathcal{L}^\perp$ is not used in our analysis, we omit it here and only describe the inherited part.

By Lemma~\ref{lem: eigenvalue_expansion} and Eq.~(\ref{eq: eigenvector_wsym}), the corresponding eigenvector $f_{\lambda,s}^{(\epsilon)}\in \mathrm{Ker}(W_{\mathrm{sym}}(\epsilon)-\nu_{\lambda,s}(\epsilon) I)$ can be expanded as
\begin{align}
    \label{eq: eigenvector_by_f}
    f_{\lambda}^{(\epsilon)} = D_{\epsilon}^{1/2}K_{\pi}R^{1/2}\tilde{f}_{\lambda,s} + O(\epsilon),
\end{align}
where $\tilde{f}_{\lambda,s} \in \mathrm{Ker}(L_{\mathrm{sym}}-\lambda I)$. 
By substituting this asymptotic expansion into the spectral mapping theorem (Proposition~\ref{prop: SMT}), we obtain the explicit asymptotic representation of the eigenvectors of the unitary operator $U(\epsilon)$ in terms of the Laplacian $L_{\mathrm{sym}}$. 
This procedure establishes the following lemma.
\begin{lemma}
\label{lem: qw_eigenvector}
Let $\lambda\in\operatorname{Spec}(L_{\mathrm{sym}})$ and let
$\tilde f_{\lambda,s}\in\operatorname{Ker}(L_{\mathrm{sym}}-\lambda I)$,
$s=1,\ldots,m_\lambda$, be the normalized eigenvectors chosen in
Lemma~\ref{lem: eigenvalue_expansion}.

For $\lambda\neq0$, the corresponding eigenvectors
$\eta_{\sigma,\lambda,s}^{(\epsilon)}\in\mathrm{Ker}\left(U(\epsilon)-e^{\sigma i\theta_{\lambda,s}(\epsilon)}I\right)$,
$\sigma\in\{\pm1\}$, satisfy
\begin{align*}
\eta_{\sigma,\lambda,s}^{(\epsilon)}(a)
=
\begin{dcases}
-\dfrac{\sigma i}
{\sqrt{2|A_{\pi(t(a))}|}}
\tilde f_{\lambda,s}(\pi(t(a)))
+O(\sqrt{\epsilon}),
& a\notin B,\\[2mm]
\dfrac{1}{2\sqrt{\lambda}}
\left[
(R^{1/2}\tilde f_{\lambda,s})(\pi(t(a)))
-
(R^{1/2}\tilde f_{\lambda,s})(\pi(o(a)))
\right]
+O(\sqrt{\epsilon}),
& a\in B.
\end{dcases}
\end{align*}
for each $a\in A$.

For $\lambda=0$, the corresponding eigenvector
$\eta_1^{(\epsilon)}\in \mathrm{Ker}(U(\epsilon)-I)$ satisfies
\begin{align*}
\eta_1^{(\epsilon)}(a)
=
\begin{dcases}
\dfrac{1}{\sqrt{|A_{\pi(t(a))}|}}
\tilde f_0(\pi(t(a)))
+O(\sqrt{\epsilon}),
& a\notin B,\\[2mm]
O(\sqrt{\epsilon}),
& a\in B.
\end{dcases}
\end{align*}
\end{lemma}
\begin{proof}
Let $\tilde{f}_{\lambda,s} \in \mathrm{Ker}(L_{\mathrm{sym}}-\lambda I)$ be an eigenvector normalized as $\|\tilde{f}_{\lambda,s}\|=1$. 
For $\lambda=0$, by Proposition \ref{prop: SMT} and Eq.~\eqref{eq: eigenvector_by_f}, we see
\begin{align*}
    \eta_{1}^{(\epsilon)}(a)
        &=\sqrt{p_{\epsilon}(a^{-1})}f_{0}^{(\epsilon)}(t(a)) \\
        &=\sqrt{p_{\epsilon}(a^{-1})}((D_{\epsilon}^{1/2}K_{\pi}R^{1/2}\tilde{f}_{0})(t(a))+O(\epsilon)) \\
        &=\sqrt{\frac{w_{\epsilon}(a)}{m_{\epsilon}(t(a))}}\left(\sqrt{m_{\epsilon}(t(a))}(R^{1/2}\tilde{f}_{0})(\pi(t(a)))+O(\epsilon) \right) \\
        &=\sqrt{w_{\epsilon}(a)}(R^{1/2}\tilde{f}_{0})(\pi(t(a)))+O(\epsilon).
\end{align*}
Here, we use the fact that
\begin{align*}
    \sqrt{\frac{w_{\epsilon}(a)}{m_{\epsilon}(t(a))}}\times O(\epsilon)=O(\epsilon).
\end{align*}
for $a \in A$.
Thus, we obtain
\begin{align*}
    \eta_{1}^{(\epsilon)}(a)=
    \begin{dcases}
        (R^{1/2}\tilde{f}_{0})(\pi(t(a)))+O(\epsilon) &: a \notin B, \\
        O(\sqrt{\epsilon}) &: a \in B.
    \end{dcases}
\end{align*}
Next we consider $\lambda \neq 0$ case.
\begin{align*}
    &\eta_{\sigma, \lambda,s}^{(\epsilon)}(a)\\
    &= \frac{1}{\sqrt{2}|\sin \theta_{\lambda,s}(\epsilon)|}\left(\sqrt{p_{\epsilon}(a^{-1})} f_{\lambda,s}^{(\epsilon)}(t(a))-e^{\sigma i \theta_{\lambda,s}(\epsilon)} \sqrt{p_{\epsilon}(a)} f_{\lambda,s}^{(\epsilon)}(o(a)) \right) \\
    &= \frac{1}{\sqrt{2}|\sin \theta_{\lambda,s}(\epsilon)|}\left(\sqrt{\frac{w_{\epsilon}(a)}{m_{\epsilon}(t(a))}}\sqrt{m_{\epsilon}(t(a))}(R^{1/2}\tilde{f}_{\lambda,s})(\pi(t(a))) \right.\\
    & \qquad \qquad \left. -e^{\sigma i \theta_{\lambda,s}(\epsilon)} \sqrt{\frac{w_{\epsilon}(a)}{m_{\epsilon}(o(a))}}\sqrt{m_{\epsilon}(o(a))}(R^{1/2}\tilde{f}_{\lambda,s})(\pi(o(a))) +O(\epsilon) \right)  \\
    &= \frac{\sqrt{w_{\epsilon}(a)}}{\sqrt{2}|\sin \theta_{\lambda,s}(\epsilon)|}\left((R^{1/2}\tilde{f}_{\lambda,s})(\pi(t(a)))-e^{\sigma i \theta_{\lambda,s}(\epsilon)}(R^{1/2}\tilde{f}_{\lambda,s})(\pi(o(a)))+O(\epsilon) \right).
\end{align*}
Here, from $\cos \theta_{\lambda,s}(\epsilon)=1-\epsilon \lambda+O(\epsilon^{2})$, we have
\begin{align*}
    e^{\sigma i \theta_{\lambda,s}(\epsilon)}&= 1+\sigma i \theta_{\lambda,s}(\epsilon) +O(\theta_{\lambda}(\epsilon)^{2})\\
    &= 1+\sigma i \sqrt{2 \lambda \epsilon} +O(\epsilon),\\
    \sin \theta_{\lambda,s}(\epsilon)&=\sqrt{2\epsilon\lambda}+O(\epsilon^{3/2}).
\end{align*}
Remark that for $a \notin B$, $\pi(t(a))=\pi(o(a))$ and $w_{\epsilon}(a)=1$  hold. 
Then, for $a \notin B$, we obtain
\begin{align*}
    &\eta_{\sigma, \lambda,s}^{(\epsilon)}(a)\\
    &= \frac{\sqrt{w_{\epsilon}(a)}}{\sqrt{2}|\sin \theta_{\lambda,s}(\epsilon)|}\left((R^{1/2}\tilde{f}_{\lambda,s})(\pi(t(a)))-e^{\sigma i \theta_{\lambda,s}(\epsilon)}(R^{1/2}\tilde{f}_{\lambda,s})(\pi(o(a)))+O(\epsilon) \right) \\
    &=\frac{1}{2\sqrt{\epsilon\lambda}+O(\epsilon^{3/2})}\left((R^{1/2}\tilde{f}_{\lambda,s})(\pi(t(a)))-(1+\sigma i \sqrt{2 \lambda \epsilon} +O(\epsilon))(R^{1/2}\tilde{f}_{\lambda,s})(\pi(t(a)))+O(\epsilon) \right) \\
    &=-\frac{\sigma i }{\sqrt{2}}(R^{1/2}\tilde{f}_{\lambda,s})(\pi(t(a)))+O(\sqrt{\epsilon}) \\
    &=-\frac{\sigma i }{\sqrt{2|A_{\pi(t(a))}|}}\tilde{f}_{\lambda,s}(\pi(t(a)))+O(\sqrt{\epsilon})
\end{align*}
For $a \in B$, recalling $w_{\epsilon}(a)=\epsilon$, we obtain
\begin{align*}
    &\eta_{\sigma, \lambda,s}^{(\epsilon)}(a)\\
    &= \frac{\sqrt{w_{\epsilon}(a)}}{\sqrt{2}|\sin \theta_{\lambda,s}(\epsilon)|}\left((R^{1/2}\tilde{f}_{\lambda,s})(\pi(t(a)))-e^{\sigma i \theta_{\lambda,s}(\epsilon)}(R^{1/2}\tilde{f}_{\lambda,s})(\pi(o(a)))+O(\epsilon) \right) \\
    &=\frac{\sqrt{\epsilon}}{2\sqrt{\epsilon\lambda}+O(\epsilon^{3/2})}\left((R^{1/2}\tilde{f}_{\lambda,s})(\pi(t(a)))-(1+O(\sqrt{\epsilon}))(R^{1/2}\tilde{f}_{\lambda,s})(\pi(o(a)))+O(\epsilon) \right) \\
    &=\frac{1}{2\sqrt{\lambda}}\left((R^{1/2}\tilde{f}_{\lambda,s})(\pi(t(a)))-(R^{1/2}\tilde{f}_{\lambda,s})(\pi(o(a))) \right) +O(\sqrt{\epsilon}) \\
\end{align*}
\end{proof}

\subsection{Proof of Theorem \ref{thm:prob_matrix}}
\begin{proof}
In the preceding discussions, we have investigated the specific eigenvalues and eigenvectors that emerge from the splitting of the eigenvalue $1$ with respect to $\epsilon$. 
As the following lemma shows, it suffices to consider only these eigenvalues and eigenvectors, since they account for almost the overlap with the initial state.
\begin{lemma}
\label{lem: total_overlap}
Let $\psi_{0}^{(j)}$ be the initial state defined in Eq.~(\ref{eq: initial}).
For each $\lambda \in \mathrm{Spec}(L_{\mathrm{sym}})\setminus\{0\}$,
let $m_{\lambda}$ denote its multiplicity.
Let $\eta_{1}^{(\epsilon)} \in \mathrm{Ker}(U(\epsilon)-I)$ and
\begin{align*}
\eta_{\sigma,\lambda,s}^{(\epsilon)}
\in
\mathrm{Ker}
\left(
U(\epsilon)
-
e^{\sigma i\theta_{\lambda,s}(\epsilon)}I
\right),
\qquad
\sigma\in\{\pm1\},
\quad
s=1,\ldots,m_{\lambda},
\end{align*}
be the corresponding normalized eigenvectors of $U(\epsilon)$.
Then, we have
\begin{align*}
&
\left|
\langle
\eta_{1}^{(\epsilon)},
\psi_{0}^{(j)}
\rangle
\right|^{2}
+
\sum_{\lambda\in\mathrm{Spec}(L_{\mathrm{sym}})\setminus\{0\}}
\sum_{s=1}^{m_{\lambda}}
\sum_{\sigma\in\{\pm1\}}
\left|
\langle
\eta_{\sigma,\lambda,s}^{(\epsilon)},
\psi_{0}^{(j)}
\rangle
\right|^{2}
=
1+O(\sqrt{\epsilon}).
\end{align*}
\end{lemma}

\begin{proof}
    By direct calculation, we get 
    \begin{align*}
        \langle \eta_{1}^{(\epsilon)}, \psi_0^{(j)} \rangle
        &= \sum_{a \in A_{j}} \overline{\left(|A_{\pi(t(a))}|^{-1/2}\tilde{f}_{0}(\pi(t(a)))+O(\sqrt{\epsilon})\right)}\times |A_{j}|^{-1/2} \\
        &= \tilde{f}_{0}(V_{j})+O(\sqrt{\epsilon}).
    \end{align*}
    \begin{align*}
        \langle \eta_{\sigma, \lambda,s}^{(\epsilon)}, \psi_0^{(j)} \rangle
        &=\sum_{a \in A_{j}} \overline{\left(-\sigma i (2|A_{\pi(t(a))}|)^{-1/2}\tilde{f}_{\lambda,s}(\pi(t(a)))+O(\sqrt{\epsilon})\right)}\times |A_{j}|^{-1/2} \\
        &=\frac{\sigma i}{\sqrt{2}}\tilde{f}_{\lambda,s}(\pi(t(a)))+O(\sqrt{\epsilon}).
    \end{align*}
    Therefore, using $\sum_{\lambda }\sum_{s} \tilde{f}_{\lambda,s}(V_j)^{2} = 1$, the total sum is given by
    \begin{align*}
        &\left|
\langle
\eta_{1}^{(\epsilon)},
\psi_{0}^{(j)}
\rangle
\right|^{2}
+
\sum_{\lambda\in\mathrm{Spec}(L_{\mathrm{sym}})\setminus\{0\}}
\sum_{s=1}^{m_{\lambda}}
\sum_{\sigma\in\{\pm1\}}
\left|
\langle
\eta_{\sigma,\lambda,s}^{(\epsilon)},
\psi_{0}^{(j)}
\rangle
\right|^{2} \\
        &= \tilde{f}_{0}(V_{j})^{2} + \sum_{\lambda \in \mathrm{Spec}(L_{\mathrm{sym}}) \setminus \{0\}}\sum_{s=1}^{m_{\lambda}} \tilde{f}_{\lambda,s}(V_j)^{2} + O(\sqrt{\epsilon}) \\
        &= \sum_{\lambda ,s} \tilde{f}_{\lambda,s}(V_j)^{2} + O(\sqrt{\epsilon}) \\
        &= 1 + O(\sqrt{\epsilon}).
    \end{align*}
    Hence, we get the desired result.
\end{proof}

From Lemma \ref{lem: total_overlap}, the probability $\mu_{t}^{(\epsilon)}(j, k)$ is given by
\begin{align*}
    &\mu_{t}^{(\epsilon)}(j, k) \\
    &= \sum_{a \in A_{k}}\left|\psi_{t}^{(j)}(a)\right|^{2} \\
    &= \sum_{a \in A_{k}}\left|\sum_{{\nu \in \mathrm{Spec}(U(\epsilon))}}\nu^{t} \eta_{\nu}(a)\braket{\eta_{\nu}, \psi_{0}^{(j)}} \right|^{2} \\
    &= \sum_{a \in A_{k}}\left| \eta_{1}^{(\epsilon)}(a) \langle \eta_{1}^{(\epsilon)}, \psi_0^{(j)} \rangle + \sum_{\lambda \in \mathrm{Spec}(L_{\mathrm{sym}}) \setminus \{0\}} \sum_{\sigma \in \{\pm1\}} e^{\sigma i t \theta_{\lambda,s}(\epsilon)} \eta_{\sigma, \lambda,s}^{(\epsilon)}(a) \langle \eta_{\sigma, \lambda,s}^{(\epsilon)}, \psi_0^{(j)} \rangle \right|^{2} + O(\sqrt{\epsilon}), \\
\end{align*}
where $\eta_{\nu}$ is the eigenvector associated with the eigenvalue 
$\nu$. Here, for $a\in A_k$, since $a\notin B$, we have $\pi(t(a))=V_k$ and $A_{\pi(t(a))}=A_k$. Thus, we see
\begin{align*}
    \eta_{1}^{(\epsilon)}(a)\langle \eta_{1}^{(\epsilon)}, \psi_0^{(j)} \rangle 
    &=\left(\frac{1}{\sqrt{|A_{k}|}}\tilde{f}_{0}(V_{k})+O(\sqrt{\epsilon})\right)(\tilde{f}_{0}(V_{j})+O(\sqrt{\epsilon})) \\
    &=\frac{1}{\sqrt{|A_{k}|}}\tilde{f}_{0}(V_{j})\tilde{f}_{0}(V_{k})+O(\sqrt{\epsilon})).
\end{align*}
\begin{align*}
&e^{it \theta_{\lambda,s}(\epsilon)}\eta_{+,\lambda,s}^{(\epsilon)}(a)\langle \eta_{+,\lambda,s}^{(\epsilon)},\psi_{0}^{(j)}\rangle+e^{-it \theta_{\lambda,s}(\epsilon)}\eta_{-,\lambda,s}^{(\epsilon)}(a)\langle\eta_{-,\lambda,s}^{(\epsilon)},\psi_{0}^{(j)}\rangle \\
&= e^{it \theta_{\lambda,s}(\epsilon)} \left(-\frac{i}{\sqrt{2|A_{k}|}}\tilde{f}_{\lambda,s}(V_{k})+O(\sqrt{\epsilon})\right) \left(\frac{i}{\sqrt{2}}\tilde{f}_{\lambda,s}(V_{j})+O(\sqrt{\epsilon})\right) \\
&\quad + e^{-it \theta_{\lambda,s}(\epsilon)} \left(\frac{i}{\sqrt{2|A_{k}|}}\tilde{f}_{\lambda,s}(V_{k})+O(\sqrt{\epsilon})\right) \left(-\frac{i}{\sqrt{2}}\tilde{f}_{\lambda,s}(V_{j})+O(\sqrt{\epsilon})\right) \\
&= \frac{1}{2\sqrt{|A_{k}|}}\tilde{f}_{\lambda,s}(V_{j})\tilde{f}_{\lambda,s}(V_{k})(e^{it \theta_{\lambda,s}(\epsilon)}+e^{-it \theta_{\lambda,s}(\epsilon)}) +O(\sqrt{\epsilon})\\
&= \frac{1}{\sqrt{|A_{k}|}}\tilde{f}_{\lambda,s}(V_{j})\tilde{f}_{\lambda,s}(V_{k})\cos(t\theta_{\lambda,s}(\epsilon))+O(\sqrt{\epsilon}).
\end{align*}

Therefore, we obtain
\begin{align*}
&\mu_{t}^{(\epsilon)}(j,k) \\
&= \sum_{a \in A_{k}}\left| \eta_{1}^{(\epsilon)}(a) \langle \eta_{1}^{(\epsilon)}, \psi_0^{(j)} \rangle + \sum_{\lambda \in \mathrm{Spec}(L_{\mathrm{sym}}) \setminus \{0\}} \sum_{s=1}^{m_\lambda} \sum_{\sigma \in \{\pm1\}} e^{\sigma i t \theta_{\lambda,s}(\epsilon)} \eta_{\sigma,\lambda,s}^{(\epsilon)}(a) \langle \eta_{\sigma,\lambda,s}^{(\epsilon)}, \psi_0^{(j)} \rangle \right|^{2} + O(\sqrt{\epsilon}) \\
&= \sum_{a \in A_{k}}\left| \sum_{\lambda \in \mathrm{Spec}(L_{\mathrm{sym}})} \sum_{s=1}^{m_\lambda}\frac{1}{\sqrt{|A_{k}|}}\tilde{f}_{\lambda,s}(V_{j})\tilde{f}_{\lambda,s}(V_{k})\cos (t\theta_{\lambda,s}(\epsilon)) \right|^{2} + O(\sqrt{\epsilon}) \\
&= \sum_{a \in A_{k}} \frac{1}{|A_k|} \left(\sum_{\lambda \in \mathrm{Spec}(L_{\mathrm{sym}})} \sum_{s=1}^{m_\lambda}\tilde{f}_{\lambda,s}(V_j)\tilde{f}_{\lambda,s}(V_k)\cos (t\theta_{\lambda,s}(\epsilon))\right)^{2} + O(\sqrt{\epsilon}) \\
&= \left(\sum_{\lambda \in \mathrm{Spec}(L_{\mathrm{sym}})} \sum_{s=1}^{m_\lambda}\tilde{f}_{\lambda,s}(V_j)\tilde{f}_{\lambda,s}(V_k)\cos (t\theta_{\lambda,s}(\epsilon))\right)^{2} + O(\sqrt{\epsilon}).
\end{align*}
Here, using $\theta_{\lambda,s}(\epsilon)=\sqrt{2\epsilon\lambda}+O(\epsilon^{3/2})$, we see
\begin{align*}
    \cos (t\theta_{\lambda,s}(\epsilon))
    &=\cos(t\sqrt{2\epsilon \lambda})+O(t\epsilon^{3/2}).
\end{align*}
Thus, the probability matrix $\mu_{t}^{(\epsilon)}(j,k)$ is given by
\begin{align}
\label{eq: order_t}
    &\mu_{t}^{(\epsilon)}(j,k) \notag\\
    &= \left(\sum_{{\lambda \in \mathrm{Spec}(L_{\mathrm{sym}})}}\sum_{s=1}^{m_\lambda}\tilde{f}_{\lambda,s}(V_j)\tilde{f}_{\lambda,s}(V_k)\cos\left(t\sqrt{2\epsilon \lambda }\right)\right)^{2} + O(\sqrt{\epsilon}+t\epsilon^{3/2})\notag\\
    &=\left( \left[\cos \left(t \sqrt{2 \epsilon L_{\mathrm{sym}}} \right)\right]_{j,k} \right)^{2} + O(\sqrt{\epsilon}+t\epsilon^{3/2}).
\end{align}
Hence, for fixed $\tau>0$, we get
\begin{align*}
    \mu_{\lfloor \tau/ \sqrt{2\epsilon}\rfloor}^{(\epsilon)}(j,k)
    =\left( \left[\cos \left(\tau \sqrt{  L_{\mathrm{sym}}} \right)\right]_{j,k} \right)^{2} + O(\sqrt{\epsilon}).
\end{align*}
Therefore, we obtain the statement of Theorem \ref{thm:prob_matrix}.
\end{proof}

\section{Summary and discussion}
\label{subsec_summary}
In this paper, we studied a discrete-time quantum walk on a finite weighted graph in which the edges belonging to an edge cut set are assigned a small weight $\epsilon>0$. When $\epsilon=0$, transitions across these weak edges are completely forbidden, and the graph separates into the connected components obtained by removing them. This setting extends our previous two-component model connected by a single weak edge to a general network of weakly connected components \cite{hosaka2026}.

By contracting each connected component into a single vertex and retaining the weak edges between the components, we introduced a reduced graph. We showed that, on the time scale $t=\Theta(\epsilon^{-1/2})$, the component-wise finding probabilities are asymptotically described by a continuous-time wave equation on this reduced graph. The equation is governed by the symmetric weighted Laplacian $L_{\mathrm{sym}}$, which is determined by the numbers of weak edges connecting the components and the numbers of arcs contained in the individual components. Thus, although the original quantum walk evolves on the full arc space, its leading-term transport between the components is determined by an operator on the much smaller vertex space of the reduced graph.

An important consequence is that the leading-term finding probabilities are independent of the detailed internal structures of the connected components. Only the reduced graph, the multiplicities of the weak edges, and the numbers of arcs in the components enter the effective dynamics. We also showed that the resulting wave equation has the same form as Newton's equation of motion for a classical spring-mass system on the reduced graph. Under this correspondence, the number of arcs in each component plays the role of its mass, while the number of weak edges connecting two components plays the role of the spring constant between them.

These results provide a reduced description of quantum-walk transport through weak connections without resolving the detailed internal structures of the individual components. A natural direction for future work is to investigate higher-order asymptotic corrections. While the leading-term dynamics depends only on the reduced graph and the component sizes, the detailed internal structures of the components are expected to appear in higher-order terms. It would also be interesting to clarify how graph connectivity, particularly the structure of edge cuts and the connections between components, influences quantum-walk transport. Such an investigation may lead to conditions under which the walker becomes strongly concentrated in a prescribed component, with possible connections to quantum search and state-transfer problems.

\section*{Acknowledgment}
T. H. acknowledges financial support from the Grant-in-Aid for JSPS Fellows (Grant No. 26KJ1214).
E. S. acknowledges financial support from the Grant-in-Aid of Scientific Research (C) Japan Society for the Promotion of Science (Grant No. 24K06863).
The authors have no competing interests to declare that are relevant to the content of this article. All data generated or analyzed during this study are included in this published article.

\bibliographystyle{abbrv}
\bibliography{ref}

@article{wong2016,
  title      = {Engineering the success of quantum walk search using weighted graphs},
  author     = {Thomas G. Wong and Pascal Philipp},
  year       = {2016},
  doi        = {10.1103/PhysRevA.94.022304},
  url        = {https://doi.org/10.1103/physreva.94.022304},
  issue      = {2},
  journal    = {Physical Review A},
  publisher  = {American Physical Society (APS)},
  volume     = {94}
}

@article{mukai2020,
  title      = {Discrete-time quantum walk on complex networks for community detection},
  author     = {Kanae Mukai and Naomichi Hatano},
  year       = {2020},
  doi        = {10.1103/PhysRevResearch.2.023378},
  url        = {https://doi.org/10.1103/physrevresearch.2.023378},
  issue      = {2},
  journal    = {Physical Review Research},
  publisher  = {American Physical Society (APS)},
  volume     = {2}
}

@article{tsomokos2011,
  title      = {Quantum walks on complex networks with connection instabilities and community structure},
  author     = {Dimitris I. Tsomokos},
  year       = {2011},
  doi        = {10.1103/PhysRevA.83.052315},
  url        = {https://doi.org/10.1103/physreva.83.052315},
  issue      = {5},
  journal    = {Physical Review A},
  publisher  = {American Physical Society (APS)},
  volume     = {83}
}

@article{wong2017,
  title      = {Coined quantum walks on weighted graphs},
  author     = {Thomas G Wong},
  year       = {2017},
  doi        = {10.1088/1751-8121/aa8c17},
  url        = {https://doi.org/10.1088/1751-8121/aa8c17},
  issue      = {47},
  journal    = {Journal of Physics A: Mathematical and Theoretical},
  pages      = {475301},
  publisher  = {IOP Publishing},
  volume     = {50}
}

@article{godsil2019,
  title      = {Discrete-time quantum walks and graph structures},
  author     = {Chris Godsil and Hanmeng Zhan},
  year       = {2019},
  doi        = {10.1016/j.jcta.2019.05.003},
  url        = {https://doi.org/10.1016/j.jcta.2019.05.003},
  journal    = {Journal of Combinatorial Theory, Series A},
  pages      = {181-212},
  publisher  = {Elsevier BV},
  volume     = {167}
}

@book{GodsilHanmengbook,
  title    = {Discrete Quantum Walks on Graphs and Digraphs},
  author     = {C. Godsil and H. Zhan},
  year       = {2023},
  isbn       = {978-1-009-26168-5},
  publisher  = {Cambridge University Press}
}

@article{AdaChanHanmeng,
  title      = {Pretty good state transfer in discrete-time quantum walks},
  author     = {A. Chan and H. Zhan},
  year       = {2023},
  issue      = {16},
  journal    = {Jounal Physics A: Mathematical and Theoretical},
  volume     = {56}
}

@article{grimmett2004,
  title      = {Weak limits for quantum random walks},
  author     = {Geoffrey Grimmett and Svante Janson and Petra F. Scudo},
  year       = {2004},
  doi        = {10.1103/PhysRevE.69.026119},
  url        = {https://doi.org/10.1103/physreve.69.026119},
  issue      = {2},
  journal    = {Physical Review E},
  publisher  = {American Physical Society (APS)},
  volume     = {69}
}

@article{ambainis2005,
  title      = {Coins make quantum walks faster},
  author     = {Andris Ambainis and Julia Kempe and Alexander Rivosh},
  year       = {2005},
  journal    = {Proceedings of the 16th   ACM-SIAM Symposium on Discrete Algorithms},
  pages      = {1099-1108}
}

@article{shenvi2003,
  title      = {Quantum random-walk search algorithm},
  author     = {Neil Shenvi and Julia Kempe and K. Birgitta Whaley},
  year       = {2003},
  doi        = {10.1103/PhysRevA.67.052307},
  url        = {https://doi.org/10.1103/physreva.67.052307},
  issue      = {5},
  journal    = {Physical Review A},
  publisher  = {American Physical Society (APS)},
  volume     = {67}
}

@article{konno2002,
  title      = {Quantum Random Walks in One Dimension},
  author     = {Norio Konno},
  year       = {2002},
  doi        = {10.1023/A:1023413713008},
  url        = {https://doi.org/10.1023/a:1023413713008},
  issue      = {5},
  journal    = {Quantum Information Processing},
  pages      = {345-354},
  publisher  = {Springer Science and Business Media LLC},
  volume     = {1}
}

@article{higuchi2014,
  title      = {Spectral and asymptotic properties of Grover walks on crystal lattices},
  author     = {Yusuke Higuchi and Norio Konno and Iwao Sato and Etsuo Segawa},
  year       = {2014},
  doi        = {10.1016/j.jfa.2014.09.003},
  url        = {https://doi.org/10.1016/j.jfa.2014.09.003},
  issue      = {11},
  journal    = {Journal of Functional Analysis},
  pages      = {4197-4235},
  publisher  = {Elsevier BV},
  volume     = {267}
}

@book{kato1982,
  title      = {A Short Introduction to Perturbation Theory for Linear Operators},
  author     = {Tosio Kato},
  year       = {1982},
  doi        = {10.1007/978-1-4612-5700-4},
  url        = {https://doi.org/10.1007/978-1-4612-5700-4},
  address    = {New York, NY},
  publisher  = {Springer US}
}

@article{ambainis2007,
  title      = {Quantum Walk Algorithm for Element Distinctness},
  author     = {Andris Ambainis},
  year       = {2007},
  doi        = {10.1137/S0097539705447311},
  url        = {https://doi.org/10.1137/s0097539705447311},
  issue      = {1},
  journal    = {SIAM Journal on Computing},
  pages      = {210-239},
  publisher  = {Society for Industrial \& Applied Mathematics (SIAM)},
  volume     = {37}
}

@inproceedings{ambainis2001,
  title      = {One-dimensional quantum walks},
  author     = {Andris Ambainis and Eric Bach and Ashwin Nayak and Ashvin Vishwanath and John Watrous},
  year       = {2001},
  doi        = {10.1145/380752.380757},
  url        = {https://doi.org/10.1145/380752.380757},
  address    = {New York, NY, USA},
  booktitle  = {Proceedings of the thirty-third annual ACM symposium on Theory of computing},
  pages      = {37-49},
  publisher  = {ACM}
}

@inproceedings{grover1996,
  title      = {A fast quantum mechanical algorithm for database search},
  author     = {Lov K. Grover},
  year       = {1996},
  doi        = {10.1145/237814.237866},
  url        = {https://doi.org/10.1145/237814.237866},
  address    = {New York, New York, USA},
  booktitle  = {Proceedings of the twenty-eighth annual ACM symposium on Theory of computing  - STOC '96},
  pages      = {212-219},
  publisher  = {ACM Press}
}

@article{kubota2022,
  title      = {Perfect state transfer in Grover walks between states associated to vertices of a graph},
  author     = {Sho Kubota and Etsuo Segawa},
  year       = {2022},
  doi        = {10.1016/j.laa.2022.04.006},
  url        = {https://doi.org/10.1016/j.laa.2022.04.006},
  journal    = {Linear Algebra and its Applications},
  pages      = {238-251},
  publisher  = {Elsevier BV},
  volume     = {646}
}

@article{godsil2012,
  title      = {When can perfect state transfer occur?},
  author     = {Chris Godsil},
  year       = {2012},
  doi        = {10.13001/1081-3810.1563},
  url        = {https://doi.org/10.13001/1081-3810.1563},
  journal    = {The Electronic Journal of Linear Algebra},
  publisher  = {University of Wyoming Libraries},
  volume     = {23}
}

@article{godsil2012a,
  title      = {State transfer on graphs},
  author     = {Chris Godsil},
  year       = {2012},
  doi        = {10.1016/j.disc.2011.06.032},
  url        = {https://doi.org/10.1016/j.disc.2011.06.032},
  issue      = {1},
  journal    = {Discrete Mathematics},
  pages      = {129-147},
  publisher  = {Elsevier BV},
  volume     = {312}
}

@article{higuchi2017,
  title      = {Periodicity of the Discrete-time Quantum Walk on a Finite Graph},
  author     = {Yusuke Higuchi and Norio Konno and Iwao Sato and Etsuo Segawa},
  year       = {2017},
  doi        = {10.4036/iis.2017.A.10},
  url        = {https://doi.org/10.4036/iis.2017.a.10},
  issue      = {1},
  journal    = {Interdisciplinary Information Sciences},
  pages      = {75-86},
  publisher  = {Graduate School of Information Sciences, Tohoku University},
  volume     = {23}
}

@article{inui2005,
  title      = {Localization of multi-state quantum walk in one dimension},
  author     = {Norio Inui and Norio Konno},
  year       = {2005},
  doi        = {10.1016/j.physa.2004.12.060},
  url        = {https://doi.org/10.1016/j.physa.2004.12.060},
  journal    = {Physica A: Statistical Mechanics and its Applications},
  pages      = {133-144},
  publisher  = {Elsevier BV},
  volume     = {353}
}

@article{hosaka2026,
  title      = {Pulsation of quantum walk between two arbitrary graphs with weakly connected bridge},
  author     = {Taisuke Hosaka and Etsuo Segawa},
  year       = {2026},
  doi        = {10.1088/1402-4896/ae67b6},
  url        = {https://doi.org/10.1088/1402-4896/ae67b6},
  issue      = {19},
  journal    = {Physica Scripta},
  pages      = {195205},
  publisher  = {IOP Publishing},
  volume     = {101}
}

@article{hosaka2025,
  title      = {Pulsation of quantum walk on {J}ohnson graph},
  author     = {Taisuke Hosaka and Etsuo Segawa},
  year       = {2025},
  doi        = {10.1007/s11128-025-04996-z},
  url        = {https://doi.org/10.1007/s11128-025-04996-z},
  issue      = {11},
  journal    = {Quantum Information Processing},
  publisher  = {Springer Science and Business Media LLC},
  volume     = {24}
}

\end{document}